\PassOptionsToPackage{hyphens}{url}
\documentclass[sigconf,10pt]{acmart} 
\usepackage{balance}
\usepackage{amsmath,amssymb}
\usepackage{xspace}
\usepackage{eqnarray}
\usepackage{tikz}
\usepackage{ifthen}
\usepackage[justification=centering]{caption}
\usepackage{subcaption}
\usepackage{graphicx}
\usepackage{rotating}
\usepackage{listings}
\usepackage{booktabs}
\DeclareCaptionFont{small}{\small} 
\usepackage{algpseudocode}
\newcommand*{\audit}{Audit Join\xspace}

\newcommand*{\sparql}{\textsc{SPARQL}\xspace}
\newcommand*{\dbpedia}{DBpedia\xspace}

\newcommand{\eat}[1]{}

\def\e#1{\emph{#1}}
\def\rdftxt#1{\textsf{#1}} 
 
\def\path{$\gamma$}

\newcommand{\mparagraph}[1]{\vskip 0.4em\noindent\textbf{#1.}\,\,}
\def\partitle#1{\mparagraph{#1}}

\title{Efficiently Charting RDF}
\author{Oren Kalinsky}
\affiliation{\institution{Technion}
\city{Haifa}
\country{Israel}
\postcode{3200003}}
\author{Oren Mishali}
\affiliation{\institution{Technion}
\city{Haifa}
\country{Israel}
\postcode{3200003}}
\author{Aidan Hogan}
\affiliation{
\institution{University of Chile}
\city{Santiago}
\country{Chile}}
\author{Yoav Etsion}
\affiliation{\institution{Technion}
\city{Haifa}
\country{Israel}
\postcode{3200003}}
\author{Benny Kimelfeld}
\affiliation{\institution{Technion}
\city{Haifa}
\country{Israel}
\postcode{3200003}}

\begin{document}

\begin{abstract}
 We propose a visual query language for interactively exploring large-scale knowledge graphs. Starting from an overview, the user explores bar charts through three interactions: class expansion, property expansion, and subject/object expansion. A major challenge faced is performance: a state-of-the-art SPARQL engine may require tens of minutes to compute the multiway join, grouping and counting required to render a bar chart. A promising alternative is to apply approximation through online aggregation, trading precision for performance. However, state-of-the-art online aggregation algorithms such as Wander Join have two limitations for our exploration scenario: (1) a high number of rejected paths slows the convergence of the count estimations, and (2) no unbiased estimator exists for counts under the distinct operator. We thus devise a specialized algorithm for online aggregation that augments Wander Join with exact partial computations to reduce the number of rejected paths encountered, as well as a novel estimator that we prove to be unbiased in the case of the distinct operator. In an experimental study with random interactions exploring two large-scale knowledge graphs, our algorithm shows a clear reduction in error with respect to computation time versus Wander Join. 
\end{abstract}

\maketitle

	


    \section{Introduction}

% Graph data management -- where nodes denote entities of relevance and edges denote relations between these entities -- has enjoyed a resurgence of interest in recent years. When compared with the relational model, graphs offer a more flexible and -- in some domains -- more intuitive representation for modelling data. In terms of flexibility, graphs do not require defining a schema up-front, but rather allow for incrementally adding new nodes and edges ``on-the-fly'' as new data are encountered. In domains such as social networks, transport networks, collaboration networks, protein--protein interactions,  etc. -- where emphasis is placed on connectivity, topology and path-finding -- graphs offer a more intuitive conceptualization of relevant data than tables~\cite{AnglesABHRV17}.

A variety of prominent knowledge graphs have emerged in recent years, including DBpedia~\cite{Bizer:2009:DCP:1640541.1640848}, Freebase~\cite{BollackerEPST08}, Wikidata~\cite{VrandecicK14}, and YAGO~\cite{HoffartSBW13} covering multiple domains, LinkedGeoData~\cite{10.1007/978-3-642-04930-9_46} for geographic data, LinkedMDB~\cite{Consens08} for movie information, and LinkedSDMX~\cite{CapadisliAN15} for financial and geopolitical data. A number of companies have also announced the creation of proprietary knowledge graphs to power a variety of end-user applications, including Google,\footnote{\url{https://googleblog.blogspot.com/2012/05/introducing-knowledge-graph-things-not.html}} Microsoft,\footnote{\url{https://blogs.bing.com/search-quality-insights/2017-07/bring-rich-knowledge-of-people-places-things-and-local-businesses-to-your-apps}} Amazon\footnote{\url{https://blog.aboutamazon.com/innovation/making-search-easier}} and eBay\footnote{\url{https://www.ebayinc.com/stories/news/cracking-the-code-on-conversational-commerce/}}, among others.

Due to their scale and diversity, a major challenge faced when considering a knowledge graph is to understand what content it contains: what sorts of entities it describes, what sorts of relations are represented, how extensive the coverage of particular domains is, etc. Prominent knowledge graphs, such as DBpedia~\cite{Bizer:2009:DCP:1640541.1640848}, Freebase~\cite{BollackerEPST08}, Wikidata~\cite{VrandecicK14}, contain in the order of tens of millions of nodes and billions of edges represented using thousands of classes and properties, spanning innumerable different domains. While a variety of approaches have been proposed to summarize or profile the content of such graphs~\cite{EllefiBBDDST18}, the general trend is towards either computing statistics and summaries offline, or relying on off-the-shelf query engines.

In this paper, we propose a conceptual approach and techniques for interactive exploration of large-scale knowledge graphs through a visual query language. This query language captures user interactions that follow Shneiderman's principle for effective data visualization and exploration: ``\e{overview first, zoom and filter, then details-on-demand}''~\cite{Schneiderman96}. The result of a query is a bar chart over a set of focus nodes that are defined iteratively by the user via three interactions: \e{class expansion}, which focuses on the sub-classes of a selected class bar, \e{property expansion}, which focuses on the properties defined on instances of a class, and \e{subject/object expansion}, which focuses on entities in the source/target of a given property. At each stage, the focus nodes of the current bar chart can be filtered by a search condition. Each interaction constitutes a visual exploration step, with the sequence of interactions captured by the query language.

Given the size and diversity of prominent knowledge graphs, the number of (intermediate) results that can be generated by queries, and the goal of supporting interactive exploration, a major challenge faced is that of performance. In preliminary experiments with Virtuoso~\cite{Erling12}---a state-of-the-art SPARQL query engine---we found, for example, that computing the distribution of properties over all nodes in DBpedia takes over 5 minutes; such runtimes preclude the possibility of interactive exploration. 

To face the critical performance problem, we investigate two orthogonal approaches. First, we explore the deployment of a query engine from the recent breed of \e{worst-case optimal} join algorithms~\cite{DBLP:journals/jacm/NgoPRR18}, in order to avoid an explosion of intermediate results when processing multiway joins over large graphs; for these purposes, we select the Cache Trie Join algorithm~\cite{DBLP:conf/edbt/KalinskyEK17} to evaluate queries. Second, with the intuition that precise counts are not always required, we explore \e{online aggregation} algorithms~\cite{Hellerstein:1997:OA:253260.253291} that trade precision for performance, computing approximate counts at a fraction of the cost observed even in the worst-case optimal setting; for these purposes, we select the Wander-Join algorithm~\cite{DBLP:conf/sigmod/0001WYZ16}. In essence, Wander Join applies a random walk between database tuples that (jointly) match the join query, and upon termination, updates an estimator of the aggregate function.

Ultimately, inspired by the work of Zhao et al.~\cite{Zhao:2018:RSO:3183713.3183739}, we conclude that these two approaches are complementary. We offer an algorithm that combines online aggregation with exact computation. The general idea to apply the random walk of Wander Join, and at each step, consider replacing the remaining walk with a precise computation of the space of possible suffixes, this time using Cached Trie Join. This consideration is done via an estimate of selectivity. The estimator needs to be updated accordingly, and we prove that it remains unbiased. Furthermore, we extend our algorithm to estimate counts in the presence of the distinct operator, which is crucial to our exploration use case. We call the resulting algorithm Audit Join, and prove that it provides unbiased estimators of counts, with and without the distinct operator. In experiments that evaluate randomly-generated exploration queries over two knowledge graphs, we show that our algorithm dramatically reduces error with respect to the computation time when compared with Wander Join.

\subsection*{Contributions}
Our contributions are summarized as follows.
\begin{itemize}
\item We propose a formal model of an exploration approach over knowledge graphs.
\item We describe a system implementation of the exploration model.
\item We devise Audit Join---a specialized online-aggregation algorithm for the backend of our proposed model, and prove that it produces unbiased estimations of counts.
\item We describe an experimental study of performance over random explorations, showing the benefits of Audit Join over the state of the art.
\end{itemize}
	\section{Related Work}
We now give an overview of related work, focusing on two aspects:
exploration tools for knowledge graphs, and relevant algorithms for 
query evaluation.

\subsection*{Exploration Tools}
A variety of approaches have been proposed in recent years for
exploring and visualizing graph-structured
data~\cite{journals/semweb/DadzieR11}.
	
\paragraph{Faceted Browsing:} Among the most popular approaches that
have been studied for exploring knowledge graphs is that of
\textit{faceted browsing}, where users incrementally add restrictions
-- called \textit{facets} -- to restrict the current
results~\cite{TzitzikasMP17}. Early works mainly focused on smaller,
domain-specific graphs, among which we mention the mSpace
system~\cite{SchraefelWRS06} in the multimedia domain,
BrowseRDF~\cite{OrenDD06} in the crime domain, 
/facet~\cite{HildebrandOH06} and
Ontogator~\cite{MakelaHS06} in the art domain, or more
recently, ReVeaLD~\cite{KamdarZHDD14} in the biomedical
domain, and Hippalus~\cite{TzitzikasBPMN16} in the zoology
domain. Such works typically have dealt with smaller-scale and/or
homogeneous graphs with few classes and properties, focusing on
usability and expressivity rather than issues of scale or performance.
	
However, with the growth of large-scale multi-domain knowledge graphs
like DBpedia, Freebase or Wikidata, a number of faceted browsers have
been proposed that support thousands of classes and properties and
upwards of a hundred million edges, as needed for such datasets. Among
these systems, we can mention Neofonie~\cite{HahnBSHRBDS10},
Rhizomer~\cite{BrunettiGA13},
SemFacet~\cite{ArenasGKMZ16},
Semplore~\cite{WangLPFZTYP09} and
Sparklis~\cite{Ferre14} for exploring DBpedia;
Broccoli~\cite{BastB13} and Parallax~\cite{HuynhK09}
for exploring Freebase; and GraFa for exploring
Wikidata~\cite{MorenoH18}. Of these systems, many do not present
runtime performance
evaluation~\cite{HuynhK09,HahnBSHRBDS10,WagnerLT11,BrunettiGA13},
delegate query processing to a general-purpose query
engine~\cite{HuynhK09,HeimEZ10,BrunettiGA13,Ferre14}, apply a manual
selection of useful facets or a subset of
data~\cite{HahnBSHRBDS10,ArenasGKMZ16}, and/or otherwise rely on a
materialization approach to cache meta-data (such as
counts)~\cite{BrunettiGA13,BastB13,MorenoH18}. Compared to such works,
we focus on scalability and performance; more concretely, we propose a
novel query engine specifically optimized for the types of online
aggregation queries needed by such systems.
	
\paragraph{Graph Profiling:} While faceted browsing aims to allow
users to express and answer specific questions in an intuitive manner,
other works have focused on the problem of summarizing the content of
a large knowledge graph: to provide users insights as to what the
graph does or does not contain, what are the relationships between
entities, what are the most common types, and so
forth~\cite{EllefiBBDDST18}.
	
One approach to provide users with an overview of a knowledge graph is
to compute a \textit{graph summary} or \textit{quotient
  graph}~\cite{CebiricGM15}, which groups nodes into super-nodes,
between which the most important relations are then summarized. The
conceptual summarization can be conducted by a number of techniques,
including, for example, variations on the idea of
bisimulations~\cite{SchatzleNLP13,PicalausaFHV14,ConsensFKP15,BunemanS16},
formal concept
analysis~\cite{dAquinM11,KirchbergLTLKL12,HaceneHNV13,AlamBCN15,GonzalezH18},
semantic
types~\cite{schemamap,CampinasPCDT12,DudasSM15,FlorenzanoPRV16a},
etc. Other such works rather focus on generating statistical summaries
of large graphs, in terms of the most popular classes, properties,
etc., generating bar charts and other (possibly interactive)
visualizations~\cite{AuerDML12,AbedjanGJN14,Mihindukulasooriya15,FrischmuthMTRA15,BikakisPSS17,PrincipeSPRPM18}. While
such works tackle a variety of different use-cases using a diverse
collection of techniques, all are founded on aggregation operations
applied to nodes and relationships -- either offline or using
general-purpose query engines -- generating high-level descriptions of
the graph. The aggregation algorithms we propose are online and more
efficient than general-purpose query algorithms, and could be readily
adapted to the various use-cases explored in the aforementioned
literature.

\subsection*{Query Engines}

We provide an overview of approaches for querying knowledge graphs that most closely relate to this work.

\paragraph{SPARQL Engines:} A variety of query languages have been proposed for graphs~\cite{AnglesABHRV17}. Among these, SPARQL~\cite{SPARQL} is the standard language for querying RDF graphs, and is used, for example, by public query services over the DBpedia, LinkedGeoData, and Wikidata knowledge graphs. While several query engines support
SPARQL (e.g.,~\cite{Neumann:2008:RRE:1453856.1453927,Bishop:2011:OFS:2019470.2019472,bigdata}), in this paper we take Virtuoso~\cite{Virtuoso} as a baseline query engine given its competitiveness in a number of benchmarks~\cite{Berlin,AlucHOD14}. Virtuoso maintains clustered indexes in various (redundant) orders needed to support efficient lookups on RDF graphs; these clustered indexes support both row-wise and column-wise operations where, for example, a row-wise index can be used to find a particular row from which to start reading values from a given column. To optimize for aggregate-style queries, Virtuoso applies vectorized execution on a column represented as a compressed vector of values.

\paragraph{Worst-Case-Optimal Joins:} A number of 
 worst-case optimal join algorithms have been
developed in recent years~\cite{DBLP:conf/icdt/Veldhuizen14,
  AboKhamis:2016:FQA:2902251.2902280, Ngo:2014:BWA:2594538.2594547,
  Aberger:2017:ERE:3155316.3129246, DBLP:conf/edbt/KalinskyEK17}. These algorithms evaluate join queries with a runtime that meets the Atserias--Grohe--Marx (AGM)
bound~\cite{Atserias:2008:SBQ:1470582.1470622}, which, given a join query, provides a worst-case tight bound for the size of the output. It was shown that these algorithms are not only theoretically better than traditional approaches, they are also empirically superior on graph query patterns joining relations with low dimension~\cite{Aberger:2017:ERE:3155316.3129246,
  DBLP:conf/edbt/KalinskyEK17, Nguyen:2015:JPG:2764947.2764948}. In this paper, we will adopt Cached Trie Join -- a state-of-the-art worst-case optimal join algorithm -- and contrast it with Virtuoso for answering aggregate queries generated by our exploration system; we subsequently combine this algorithm with online aggregation techniques to trade precision for performance.
	
\paragraph{Online Aggregation:} Algorithms for online aggregation provide approximate results that converge over time to the exact aggregate queries.  Since the concept was coined by Hellerstein et al.~\cite{Hellerstein:1997:OA:253260.253291} this class of algorithms grew to support additional operators and better statistical guarantees~\cite{Haas:1997:LDC:646496.695465}, as well as distributed and parallel support~\cite{DBLP:journals/pvldb/PansareBJC11,DBLP:conf/ssdbm/QinR13, DBLP:journals/dpd/QinR14}. While the solution was originally for a single table, Haas et al.~\cite{Haas:1999:RJO:304182.304208} developed Ripple Join, an online aggregation algorithm that supports joins. More	recently, Li et al.~\cite{DBLP:conf/sigmod/0001WYZ16} have introduced the Wander Join algorithm for online aggregation over join results. Wander Join uses random walks over the indexes of the joined tables to sample results. Paths that cover all joins are considered valid sampled results, while partial paths constitute rejected samples; online aggregation can then be applied on these samples as they are collected. Wander Join has also been used as an unbiased sampling method for approximate query answering~\cite{DBLP:journals/pvldb/GalakatosCZBK17} and join-size estimation~\cite{DBLP:conf/sigmod/ChenY17}. Zhao et al.~\cite{Zhao:2018:RSO:3183713.3183739} used Wander Join to precompute initial join size estimations for the problem of uniform sampling from a join query. Their sampling algorithm uses weighted random walks where the estimation has a high confidence level and exact computations otherwise. The samples are used to improve the join size estimation for better uniform sampling. We describe Wander Join in more detail in Section~\ref{sec:engine}. 

%Load-n-Go by Procopio et al.~\cite{procopio2017load} extends Wander Join to reduce the rejection rate for highly selective filters; it uses weighted sampling by zeroing the weight of samples that were filtered in a previous random walk. 

%This method prevents Wander Join from sampling the tuples again.

In this paper, we use Cached Trie Join to reduce the rejection rate of Wander Join for selective patterns; we further prove that this strategy provides an unbiased estimator of counts with and without a distinct operator, where, to the best of our knowledge, no existing online aggregation algorithm offers unbiased estimators in the distinct case. 

	\begin{figure}[t]
\centering
\includegraphics[width=0.4\textwidth]{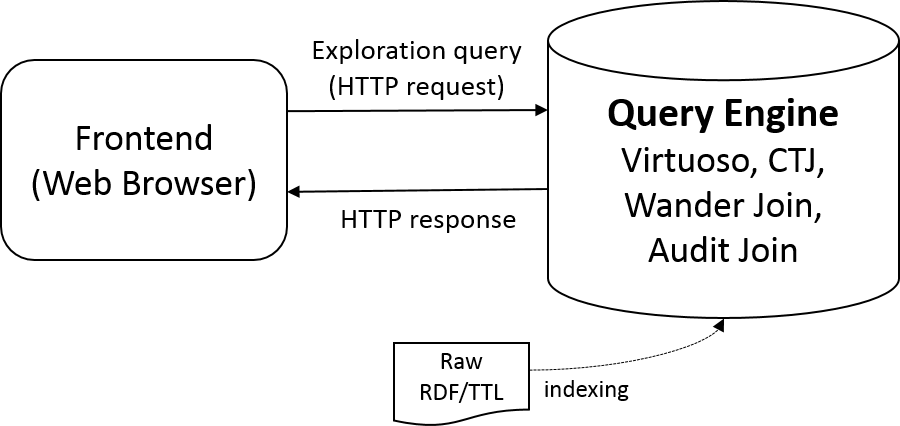}
	\caption{General system architecture}
	\label{fig:arch}
\end{figure}

\begin{figure*}[t]
	\includegraphics[width=\textwidth]{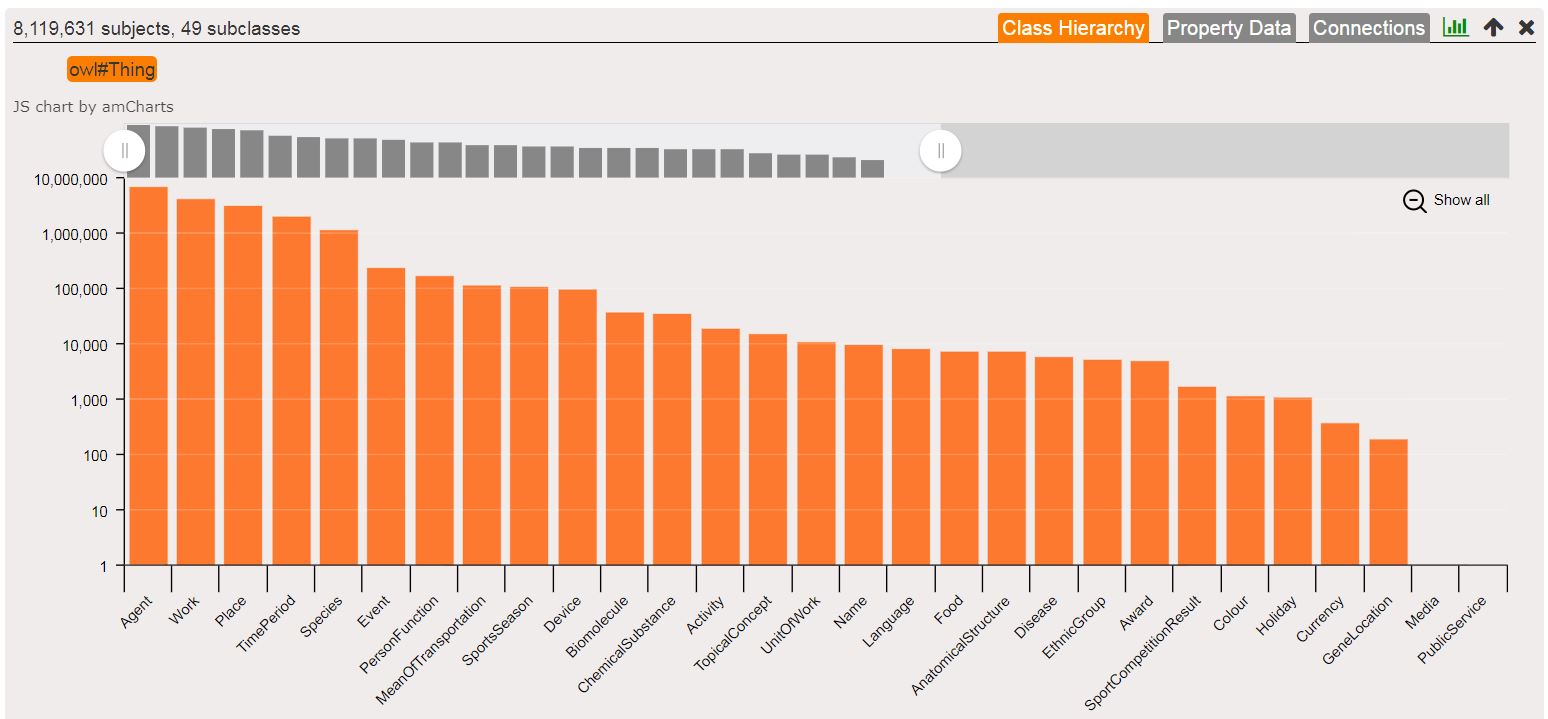}
	\caption{Initial chart in the exploration pane over \dbpedia.}
	\label{fig:dbpedia_first}
\end{figure*}
\section{System Overview} \label{sec:system} 

The focus of this work is on our formal exploration model, and its 
efficient realization as an interactive system. To give the 
intuition behind our approach, we begin with a high-level description of our
implemented system;\footnote{The system
  has been demonstrated and described in a demo paper~\cite{elinda}.} in the next sections, we delve 
into the formal model and algorithms. 
Our system offers online exploration of large-scale knowledge graphs and is implemented as a Web application that communicates with a
specialized query engine, as illustrated in Figure~\ref{fig:arch}. 

Currently our system supports exploration of an RDF graph: a set of RDF triples of the form $(s,p,o)$ where $s$ is called the \textit{subject}, $p$ is called the \textit{predicate}, and $o$ is called the \textit{object}. An RDF graph can thus be considered as a directed edge-labeled graph in which each triple encodes an edge $s \xrightarrow{p} o$. Nodes in this RDF graph may be instances of \textit{classes} (e.g. \rdftxt{Person}, \rdftxt{Movie}, etc.) where these classes may be further organized into a \textit{subclass hierarchy} (e.g., defining \rdftxt{Movie} to be a subclass of \rdftxt{Work}). We further refer to terms used in the predicate position of a triple (e.g., \rdftxt{director}) as \textit{properties}. We will provide more detail on RDF graphs in Section~\ref{sec:engine}.

During exploration of the graph, the Web application generates queries that are sent via HTTP to the backend. In this backend, the system can use any query engine that supports aggregate queries over graph patterns (more specifically, count and count--distinct operations over groupings of results for join queries). Our system currently implements four alternative query engines, which we describe in more detail in
Section~\ref{sec:engine}: Virtuoso~\cite{Virtuoso}, Cache
Trie Join (CTJ)~\cite{DBLP:conf/edbt/KalinskyEK17}, Wander
Join~\cite{DBLP:conf/sigmod/0001WYZ16}, and \audit---our bespoke
algorithm for online aggregation that we describe in Section~\ref{sec:engine}. %In any case, the setup of the query engine involves indexing raw TTL/RDF files of the explored dataset.

The user experience is visual, and no \sparql knowledge is required
from the user. In principle, the user should have only a basic
understanding of what classes and properties are. Next, we provide an overview of the user interface, and
illustrate the concept and functionality of our system through an example
exploration over the \dbpedia
dataset~\cite{Bizer:2009:DCP:1640541.1640848}.

\subsection{User Interface}
%\elinda's
The basic UI component is a \e{tabbed pane}, as illustrated in
Figure~\ref{fig:dbpedia_first}. Each tab in the pane presents a
specific \e{bar chart}, which is the result of an \e{expansion}
applied on a bar of a previous pane. (We present our formal
exploration model in Section~\ref{sec:model}.)  The tab in
Figure~\ref{fig:dbpedia_first} shows the initial bar chart for
\dbpedia. The bar chart visualizes the distribution of all \dbpedia
subjects (instances of class \rdftxt{owl:Thing}) into subclasses. Each
bar matches a specific subclass, with a height proportional to its
number of instances. The bars are sorted by decreasing
height. Hovering over a bar opens a pop-up box with basic information
such as the number of instances, and the number of direct and indirect
subclasses. To support visualization of large charts with many classes
or properties, a widget allows to control the visible part of the
chart.

The user can then navigate down the class hierarchy in order to focus
the exploration on a set of instances in a class of interest. Class navigation is done by
clicking a bar, which opens a new pane under the current one. (In cases where top-down class navigation is less intuitive, our system offers an autocomplete search box for class types, based on a list
that is populated by collecting all subjects of type
\rdftxt{owl:Class} or \rdftxt{rdfs:Class}. Selecting a class in this way immediately opens the associated pane without the need to drill down.)

A second tab in the pane shows a \e{property chart}---the result of
what we call a \e{property expansion}, as illustrated in
Figure~\ref{fig:dbpedia_props}. Here, each bar represents a specific property and the count represents distinct elements of the current focus set with some value for that property. By default, properties that emanate from the
focus set (\e{outgoing properties}) are shown, but the user may
switch to displaying the \e{incoming properties} for the current focus set. Bars are then sorted by \e{coverage}---the percentage of
focus elements that have the property as outgoing or incoming,
respectively. The number of properties may be very large, and
therefore, our system supports filtering out properties with a coverage lower
than a threshold adjustable by the user.
For example, in Figure~\ref{fig:dbpedia_props}, only 57 properties out of 722 possible
properties are shown.

Getting general statistics about the dataset and its classes is
essential, yet a user may be interested in looking into the specific
instances of the dataset as well. For that, a \e{data table} that
appears below the property chart shows the values of selected
properties (bars) in the chart. The \sparql query used to generate the
data table may be retrieved by the user for downstream
consumption. \e{Data filters} attached to table columns may restrict
the displayed data.

\begin{figure*}[t]
\centering
\includegraphics[width=0.92\textwidth]{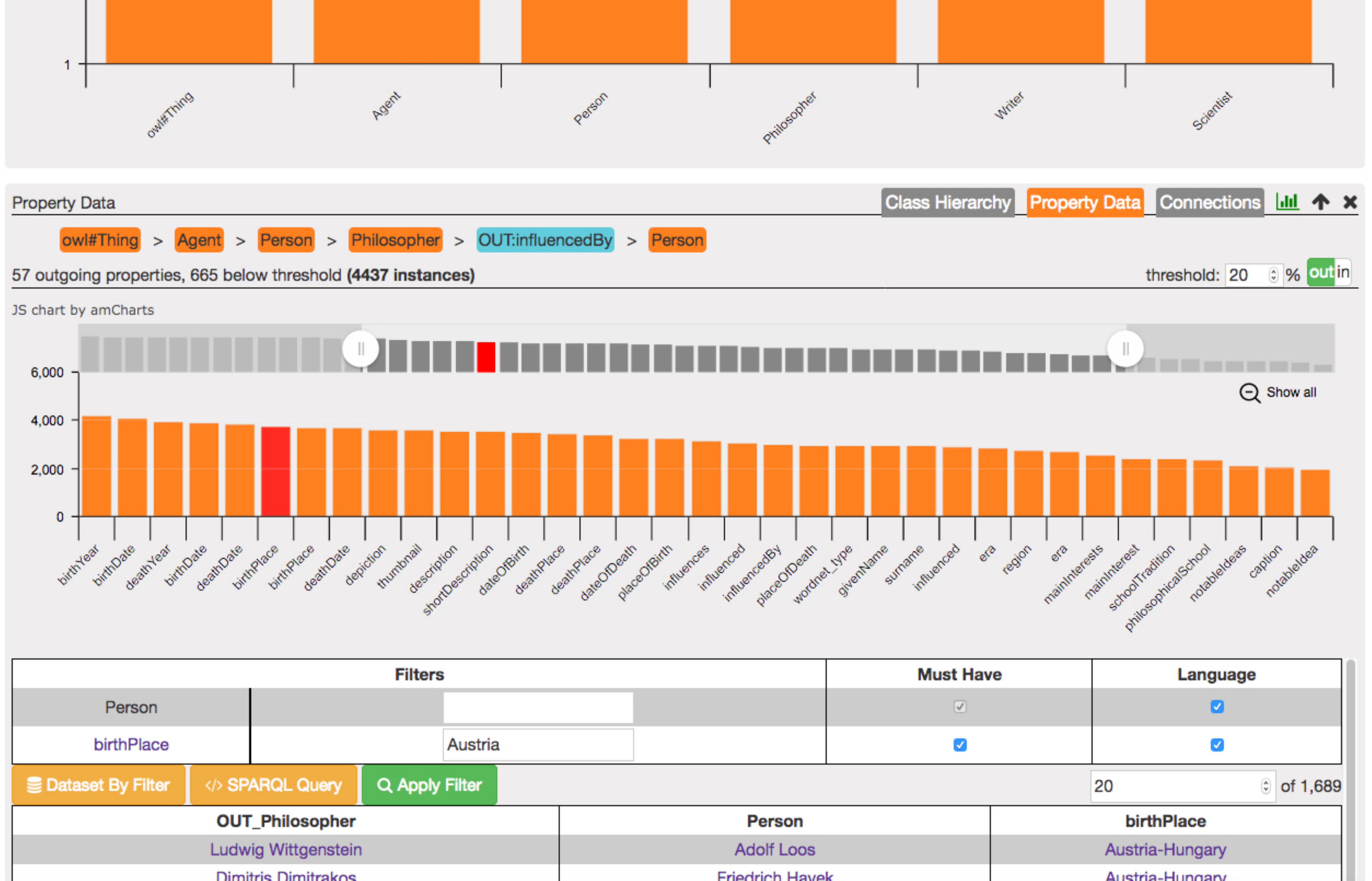}
	\caption{Two exploration panes over \dbpedia (upper is partially visible). Lower pane shows property data about \e{persons} who have influenced \e{philosophers}.}
	\label{fig:dbpedia_props}
\end{figure*}

\subsection{Illustrating Example}
To illustrate the system, consider the following scenario. Suppose
that the user is interested in philosophers, and in particular, they
wish to learn about people who have influenced philosophers.  The
exploration starts by navigating to the class \rdftxt{Philosopher}. It
is done by opening three subsequent panes: \rdftxt{Agent}
$\rightarrow$ \rdftxt{Person} $\rightarrow$
\rdftxt{Philosopher}. Then, switching to the property chart reveals
the most significant outgoing properties that philosophers have in \dbpedia,
one of them being \rdftxt{influencedBy}.
%Selecting this property and
%applying an \e{object expansion} opens a new pane, shown in
%Figure~\ref{fig:dbpedia_props}, with instances of type \rdftxt{Person}
%connected to philosophers via the \rdftxt{influencedBy} property. This
%allows the user to further explore \e{only} people who have influenced
%philosophers and not the entire \rdftxt{Person} set.
Selecting this property and
applying an \e{object expansion} opens a new pane, showing the different class types connected to philosophers via the \rdftxt{influencedBy} property. Clicking the \rdftxt{Person} type reveals the pane shown in Figure~\ref{fig:dbpedia_props}, with instances of type \rdftxt{Person} connected to philosophers (via that property). This allows the user to further explore \e{only} people who have influenced philosophers and not the entire \rdftxt{Person} set. 
Furthermore, a
\e{filter} allows to restrict the bar chart to the people born in
\e{Austria} (who have influenced philosophers), as shown in the
figure.

	\def\uris{\mathbf{U}} 
\def\lit{\mathbf{L}} 
\def\e#1{\emph{#1}}
\def\rdftxt#1{\textsf{#1}} 
\def\Qcs{Q_{\mathrm{cs}}}
\def\Qp{Q_{\mathrm{p}}} 
\def\Qc{Q_{\mathrm{c}}} 
\def\angs#1{\langle #1\rangle} 
\def\set#1{\{#1\}}\ 
\def\B{\mathbf{B}} 
\def\labels{\mathrm{labels}}
\def\init{\mathsf{init}}
\def\expansion#1{\vskip0.5em\par\noindent\underline{#1:}\,\,}

\section{Exploration Model}\label{sec:model} 
In this section, we describe our formal framework for exploring
knowledge graphs. We begin with an intuitive overview.

\subsection{General Idea}
The formal model underlying our exploration language iteratively applies the basic principle for effective data visualization by
Shneiderman~\cite{Schneiderman96} mentioned previously: ``Overview first, zoom and filter,
then details-on-demand.'' Specifically, our formal model is based on bar charts
over focus sets of nodes (URIs) that are constructed incrementally by
the user. The model is based on the following components.
\begin{itemize}
\item A \e{bar chart} consists of a set of \e{bars}, each representing
  a portion of the focus set. Figure~\ref{fig:dbpedia_first}, for
  example, depicts a bar chart in our implemented system.
\item The user selects a bar from the bar chart, and applies an
  \e{expansion operation} that transforms a bar into a new bar chart;
  the portion of the selected bar becomes the focus set of the new bar
  chart.
\item Another operation is \e{filtering} that can be applied to
  restrict the bar chart (and each bar within) according to a Boolean
  criterion over the focus set. Hence, this operation transforms one
  bar chart into another one.
\end{itemize}
The user can then continue the exploration of the new bar chart, and
hence, construct focus sets of arbitrary depths.  In what follows, we
give a formal definition of the data and exploration model.

\subsection{Formal Framework}
We now present the formal framework.
\partitle{RDF graphs} We adopt a standard model of RDF data (omitting blank nodes for brevity).
Specifically, we assume collections $\uris$ of \e{Unique Resource
  Identifiers} (URIs) and $\lit$ of \e{literals}. An \e{RDF triple},
is an element of $\uris\times\uris\times(\uris\cup\lit)$. An \e{RDF
  graph} is a finite collection $G$ of RDF triples.  In the remainder
of this section, we assume a fixed RDF graph $G$.  A URI $u$ is said
to be \e{of class} $c$ if $G$ contains the triple
$(u,\rdftxt{rdf:type},c)$. One could also define 
membership in a class by joining the \rdftxt{rdf:type} value with the \e{transitive/reflexive closure} on subclasses; the choice between the two is orthogonal to our model.

\begin{figure}[b]
\input{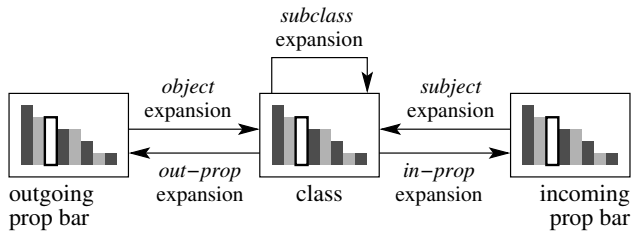}
\caption{State transitions in the exploration model\label{fig:exps}}
\end{figure}

\partitle{Bar charts}
We model the visual exploration of the RDF graph $G$ by means of bar
charts that are constructed in an iterative, interactive manner. We
have three kinds of bars:
\begin{itemize}
\item A \e{class bar} represents URIs with a common class (e.g.,
  the class $\rdftxt{Person}$).
\item An \e{outgoing-property bar}, or \e{out-property bar} for short,
  represents URIs that are the subject (source) of a common associated outgoing
  property (e.g., subjects of $\rdftxt{locatedIn}$ triples).
\item Analogously, an \e{incoming-property bar}, or \e{in-property
    bar} for short, represents URIs that are the object (target) of a
  common associated incoming property (e.g., objects of $\rdftxt{locatedIn}$
  triples).
\end{itemize}
For a bar $B$, we denote by $\uris(B)$ the set of URIs represented by
$B$. The \e{category} of $B$ is the corresponding class or property,
depending on the kind of $B$. A \e{bar chart} (or simply ``\e{chart}'' in what follows) is a mapping from
categories to bars.

\partitle{Bar expansion}
A \e{bar expansion} is a function $E$ that transforms a given bar $B$
into a chart $E(B)$.  We define specific bar expansions $E$ that are
also implemented in our system.  As a consequence, we arrive at a 
transition system between chart types, as depicted in Figure~\ref{fig:exps}.
%\begin{itemize}
%\item 
\expansion{Subclass expansion} This expansion is allowed only on class
bars $B$; in this case, the category $c$ of $B$ is a class. The
categories of the chart $E(B)$ are all the \e{subclasses} of $c$; that
is, the URIs $c'$ such that $G$ contains the triple
$(c',\rdftxt{rdfs:subClassOf},c)$.  The bar $B_{c'}$ that $E(B)$ maps
to $c'$ is a class bar with the category $c'$, and the set
$\uris(B_{c'})$ consists of all the URIs $u\in\uris(B)$ such that $u$
is of the class $c'$. 

%\item 
\expansion{Out-property expansion} This expansion is only allowed on
a class bar $B$. The categories of the chart $E(B)$ are the outgoing
properties of the URIs of $B$; that is, the URIs $p$ such that $G$
contains $(s,p,o)$ for some $o$ and $s\in\uris(B)$. The bar $B_{p}$ that
$E(B)$ maps to $p$ is an out-property bar with category $p$. The
set $\uris(B_p)$ consists of all URIs $s\in\uris(B)$ with the
property $p$; that is, all $s\in\uris(B)$ such that $(s,p,o)\in G$ for
some $o$.

%\item 
\expansion{In-property expansion} This expansion is analogous to the
out-property expansion, except that the bar $B_{p}$ mapped to $p$ by $E(B)$ is an \e{in-property} bar with the category $p$, and
$\uris(B_p)$ consists of all URIs $o\in\uris(B)$ that have the
incoming property $p$; that is, all $o\in\uris(B)$ such that
$(s,p,o)\in G$ for some $s$.

%\item 
\expansion{Object expansion} This expansion is enabled only for
out-property bars $B$; recall that, in this case, the category of $B$
is a property $p$. The categories of the chart $E(B)$ are the classes
$c$ of the \e{objects} that are connected to the URIs in $\uris(B)$
through the property $p$; that is, the classes $c$ such that for some
triple $(s,p,o)\in G$ it is the case that $s\in\uris(B)$ and $o$ is of
class $c$. The bar $B_{c}$ that $E(B)$ maps to $c$ is a class bar with
the category $c$, and $\uris(B_{c})$ consists of the $p$-targets of
type $c$; that is, all URIs $o$ of class $c$ such that $(s,p,o)\in G$
for some $s\in\uris(B)$.

%\item 
\expansion{Subject expansion} This expansion is analogous to the
object expansion, but considers the subjects of incoming properties
rather than the objects of outgoing properties. Specifically, this
expansion is enabled only for in-property bars $B$, where the
category of $B$ is a property $p$. The categories of $E(B)$ are the
classes $c$ of the \e{subjects} that are connected to the URIs in
$\uris(B)$ through the property $p$; that is, the classes $c$ such
that for some triple $(s,p,o)\in G$ it is the case that $o\in\uris(B)$
and $s$ is of class $c$. The bar $B_{c}$ that $E(B)$ maps to $c$ is a
class bar with the category $c$, and $\uris(B_{c})$ consists of the
$p$-sources of type $c$; that is, all URIs $s$ of class $c$ such that
$(s,p,o)\in G$ for some $o\in\uris(B)$.
% \end{itemize}

\partitle{Exploration}
Our model enables the exploration of $G$ by allowing the user to
construct a list of bar charts in sequence, with each successive chart exploring a bar of
the previous chart. The exploration begins with a predefined
\e{initial chart} that we denote by $\B_0$. In our implementation,
this bar has the form $E(B)$ where $E$ is the subclass expansion and
$B$ is a bar that consists of all URIs of a predefined class; a
sensible choice for this class is a general type such as
$\rdftxt{owl:Thing}$. By \e{exploration} we formally refer to a
sequence of the form
\[\B_0\mapsto(c_1,E_1)\mapsto\B_1\,,\,(c_2,E_2)\mapsto\B_2\,,\,\dots\,,\,(c_{m},E_{m})\mapsto\B_m\]
where each chart $\B_i$ is obtained by selecting the bar $B$ of
category $c_i$ from the chart $\B_{i-1}$ and applying to $B$ the
expansion $E_i$ (assuming $E_i$ is allowed on $B$).

\partitle{Filtering}
In addition to the collection of expansion operations, the model and implemented
system also allow for \e{filtering} conditions, such as ``restrict to
all URIs $s$ who were born in Africa''; that is, for some $c$ it is the
case that $(s,\rdftxt{bornIn},c)\in G$ and
$(c,\rdftxt{name},\texttt{`Africa'})\in G$. Formally, a condition is
abstracted simply as a subset $F$ of $\uris$; when applying $F$ to a
chart $\B$, the resulting chart is obtained from $\B$ by restricting
$\uris(B)$ to $F\cap\uris(B)$ for every bar $B$.

\begin{example}
For illustration, we consider another \dbpedia scenario, now expressed in the terminology of the formal exploration model. Suppose that the user is interested in understanding what information \dbpedia has on cities where scientists were born. 
The initial bar chart shown (Figure \ref{fig:dbpedia_first}) is the result of a subclass expansion applied to the bar $B$ that corresponds to 
(i.e., $\uris(B)$ consists of)  all URIs of type \rdftxt{owl:Thing}. A total of 49 top-level classes (bars) are shown, where the user may observe that the most popular classes are \rdftxt{Agent} and
\rdftxt{Work}. The user then applies a subclass expansion on the \rdftxt{Agent} bar to build a bar chart over the agents. Additional 
subclass expansions are then applied to focus on the \rdftxt{Scientist} nodes (through class \rdftxt{Person}). Next, an out-property expansion is applied to get the
distribution of outgoing properties of scientists, and from there the user selects the \rdftxt{birthPlace} bar. An object expansion over this bar results in the bar chart over the birth places of scientists, and from there the user selects the \rdftxt{City} bar. 
\qed
\end{example}

	\def\sp{\,\,}
\def\ind{\quad\quad}
\def\estwj{C_{\mathsf{wj}}}
\def\estajn{C_{\mathsf{aj}}}
\def\estajd{C_{\mathsf{aj}}^{\mathsf{d}}}
\def\eqdef{\mathrel{{:}{=}}}
\def\pr{\mathrm{Pr}}
\def\fpaths{\boldsymbol{\Gamma}}
\def\ppaths{\boldsymbol{\Delta}}

\section{Query Engine and Algorithms} \label{sec:engine}

To feature interactive exploration, the underlying query engine of the
system should answer multiway join queries in less than a
second. However, in initial experiments with the Virtuoso system, such
queries would sometimes take minutes to complete. Algorithms that
implement worst-case-optimal joins have recently been shown to be
capable of orders-of-magnitude speedup compared to traditional join
approaches~\cite{DBLP:journals/jacm/NgoPRR18,DBLP:conf/sigmod/NguyenABKNRR14},
and hence, offer a promising alternative.  Still, in experiments with
Cached Trie Join~\cite{DBLP:conf/edbt/KalinskyEK17} -- a
state-of-the-art representative of these join algorithms -- queries that
require large join results on multi-domain knowledge graphs (e.g.,
DBpedia) may still take tens of seconds to run.

With the goal of reaching acceptable performance, we turn to \e{online
  aggregation}, relaxing the expectation of exact counts to instead
aim for a fast but approximate initial response whose error reduces
over time~\cite{Hellerstein:1997:OA:253260.253291}. Such a compromise
is well justified in the context of this work, since queries are used
for rendering bar charts that can suffer loss of precision with
limited impact on the user experience.  Along these lines, we
investigate use of Wander Join~\cite{DBLP:conf/sigmod/0001WYZ16},
which is designed for aggregate queries over the grouped results of
join queries; this algorithm has been empirically demonstrated to
offer much better convergence compared to traditional online
aggregation approaches in experiments over
TPC-H~\cite{DBLP:conf/sigmod/0001WYZ16}. However Wander Join has two
limitations in our specific use-case: \e{(1)} rejected paths slow
convergence of the estimations, \e{and (2)} it does not support (i.e.,
provide an unbiased estimator) for the count-\e{distinct} operator.

We thus propose a novel online-aggregation algorithm, \audit, which
addresses these limitations of Wander Join. First, in cases where a
high number of rejected paths are deemed likely to occur, \audit
defers to partial exact computations using Cached Trie Join. Second,
\audit incorporates a novel estimator for counts under the distinct
operator that we prove to be unbiased. 

This section now discusses the various algorithms we employ to improve
query performance in our interactive exploration setting, starting
with preliminaries on query translation, then discussing Cached
Trie Join and Wander Join, before detailing our \audit proposal.

\subsection{Query Translation and Structure} \label{sec:structure}
Section~\ref{sec:model} defined our exploration model. The five
operations of subclass, in-property, out-property, object and subject
expansions are translated to SPARQL queries via our query
engine. These SPARQL queries produce the information required to
generate the next bar chart by first executing a multiway join that
encodes the expansions thus far, then a grouping on the URIs of the
next chart, and finally a distinct count on the focus set of the next
chart. Due to the structure of exploration steps, cyclic queries
cannot occur.

The general form of these SPARQL queries is illustrated by the query
template in Figure~\ref{fig:general-query}. Here, each pattern of the
form $a_i~b_i~c_i$ refers to a \e{triple pattern}, where each term
$a_i$, $b_i$ and $c_i$ (where $1 \leq i \leq n$) is either a variable
(e.g., $\texttt{?s}$) or a constant (e.g., $\texttt{<Person>}$). A
variable may appear in at most two triple patterns. Finally, $\alpha$
denotes a variable that will be assigned the URIs of the next bar
chart (either some $b_i$, or some $c_i$ where $b_i =
\texttt{rdf:type}$), while $\beta$ returns the focus set of the next
bar chart (either some $a_i$ or $c_i$). As an example, the exploration
\e{birthplaces of persons} is translated to the SPARQL query shown in
Figure~\ref{fig:specific-query}.

\lstset{mathescape,columns=fullflexible,basicstyle=\ttfamily,frame=single}
\begin{figure}
\centering
\begin{minipage}[t]{0.9\columnwidth}
\begin{lstlisting}
SELECT $\alpha$ COUNT(DISTINCT $\beta$) WHERE { 
    $a_1\ b_1\ c_1\ .$  
    $\ldots$
    $a_n\ b_n\ c_n\ .$
} GROUP BY $\alpha$
\end{lstlisting}
\end{minipage}
\caption{The general form of an exploration query\label{fig:general-query}}
\end{figure}

\begin{figure}[b]
\centering
\begin{minipage}[t]{0.9\columnwidth}
\begin{lstlisting}
SELECT ?c COUNT(DISTINCT ?o) WHERE { 
    ?s <birthPlace> ?o. 
    ?s rdf:type ?sc. 
    ?sc rdf:type <Person>.
    ?o rdf:type ?c. 
} GROUP BY ?c
\end{lstlisting}
\end{minipage}

\caption{An instance of an exploration query\label{fig:specific-query}}
\end{figure}

\newtheorem*{remark}{Remark}

\begin{remark}
  In practice, triple patterns with the ``\texttt{rdf:type}'' property are joined with the reflexive/transitive closure of subclasses. For example, in the pattern ``\texttt{?sc rdf:type
    <Person>}'' of Figure~\ref{fig:specific-query}, \texttt{?sc} will also be mapped to instances of (possibly indirect) subclasses of \texttt{<Person>}. We materialize this subclass closure and view it as a raw relation; instances, on the other hand, are typed per the original data and joined with the subclass closure at runtime. For simplicity, we leave the subclass closure implicit in the presentation of the queries since, as previously mentioned, it is orthogonal to the model.\qed
\end{remark}

In the remainder of this section, we denote by $G_i$ the subset of $G$
that consists of all the triples of the knowledge graph $G$ that match
the triple pattern $(a_i,b_i,c_i)$, where a triple $(a,b,c)$
\e{matches} $(a_i,b_i,c_i)$ if the two agree on the constants (that
is, if $a_i$ is a constant then $a=a_i$, and so on).

\subsection{Aggregation via Cached Trie Join}\label{sec:ctj}
The exact evaluation we incorporate in our approach is based on the \e{Cached Trie Join}
algorithm (CTJ)~\cite{DBLP:conf/edbt/KalinskyEK17}. This algorithm
incorporates caching of intermediate join results on top of the
\e{LeapFrog Trie Join} algorithm
(LFTJ)~\cite{DBLP:conf/icdt/Veldhuizen14}---a backtracking join
algorithm that traverses over trie indexes. In our context, we
maintain six trie indexes over $G$, each corresponding to an ordering
of the three attributes ($s$, $p$ and $o$). The trie index has a root,
and under the root a layer with the values of the first attribute, and
then a layer with the values of the second attribute, and then the
third attribute. Each triple $(s,p,o)$ corresponds to a unique
root-to-leaf path of the trie. For example, if the order is $(p,o,s)$,
then the first layer corresponds to the predicates, the second to the
objects, and the third to the subjects; in this case, a path
$\textrm{root}\rightarrow b\rightarrow c\rightarrow a$ represents the
triple $(a,b,c)$ of $G$.  In our implementation, B-tree like indexes
are used, similar to the indexes commonly used in SPARQL (and other) query engines.

LFTJ assumes a predetermined order over the variables, say
$x_1,\dots,x_m$. We access the tuples of $G_i$ using a trie $T_i$ with
an order that is consistent with the predetermined order.  For
example, if the triple $a_i$ $b_i$ $c_i$ is \texttt{?q <birthPlace>
  ?r} and \texttt{?r} precedes \texttt{?q} in the predefined order,
then $T_i$ will be the trie for $(o,s,p)$, $(o,p,s)$, or $(p,o,s)$.
LFTJ uses a backtracking algorithm that walks over the $T_i$ and looks
for assignments for $x_1,\dots,x_m$. It starts by finding the first
matching value $v_1$ for $x_1$. Then, the tries $T_i$ that contain
$x_1$ restrict their search to the subtree under $x_1=v_1$. Next, it
looks for the first match $v_2$ of the next variable $x_2$, and all
relevant tries restrict to the subtree under $x_2=v_2$.  The algorithm
continues to remaining variables, until a match is found for all
variables, or it cannot find a matching value for the next
variable. Once a match is found, or the algorithm gets stuck, it
backtracks to the next value of the previous scanned $x_i$. Grouping
and counting are applied in the straightforward manner.

While LFTJ guarantees worst-case optimality, it frequently re-computes
the same intermediate joins, since it does not materialize any of the
intermediate results~\cite{DBLP:conf/edbt/KalinskyEK17}. To
effectively reuse the partial answers, CTJ augments LTFJ with a cache
structure guided by a tree decomposition of the query, guaranteeing
the correctness of the algorithm. In the use case of this paper, the
tree decomposition is easily determined by the path formed by the
query. CTJ uses different caching schemes to cache partial count
results that are later reused. Empirically, CTJ can achieve orders of
magnitude speedup over LFTJ and other known join algorithms for
graph queries on relations with a low
dimension~\cite{DBLP:conf/edbt/KalinskyEK17}.

%We extend CTJ to support Group By aggregations and DISTINCT count. Group By is supported by fixing the caching structure in a specific column and maintaining the caching structure...

% Audit Join 
\subsection{Wander Join}
Wander Join is an online aggregation algorithm that is designed for
aggregates over
joins~\cite{DBLP:journals/jacm/NgoPRR18,DBLP:conf/sigmod/NguyenABKNRR14}.
Since Wander Join does not support the distinct operator, we ignore
the operator in this section. For presentation sake, we begin by also
ignoring the grouping operator, and assume that we only need to count
the number of matches for the variables. We discuss grouping later in
the section.

Given the query of Figure~\ref{fig:general-query}, Wander Join samples
query answers via independent random walks over the $G_i$, in contrast
to the full pre-order traversal of CTJ.  It estimates the count by
adopting the Horvitz-Thompson
estimator~\cite{horvitz1952generalization}, where each random walk
$\gamma$ produces an estimator $\estwj(\gamma)$ that we describe next,
and the final estimator is simply the average of the $\estwj(\gamma)$
over all random walks $\gamma$. The walk $\gamma$ is constructed as
follows. We first select a random tuple $t_1$ uniformly from
$G_1$. Next, select a tuple $t_2$ that is consistent with $t_1$; that
is, $t_1$ and $t_2$ agree on the common attributes; the choice is
again uniform among all consistent $t_i$.  We continue in this way,
where in the $i$\textsuperscript{th} step we select a random tuple $t_i$ from $G_i$ such
that $t_i$ is consistent with $t_{i-1}$. If, at any point, no matching
$t_i$ exists, then the random walk $\gamma$ terminates and
$\estwj(\gamma)=0$.  Otherwise, denote by $d_i$ the number of possible
ways of selecting $t_i$, for $i=1,\dots,n$. Note that the probability
of $\gamma$ is $\prod_{i=1}^n1/d_i$.  The estimator $\estwj(\gamma)$
is defined by
\[\estwj(\gamma)\eqdef\prod_{i=1}^n d_i=\frac{1}{\pr(\gamma)}\,.\]
The estimator $\estwj(\gamma)$ is unbiased; consequently, the final
estimator (i.e., the average) is also unbiased. To see why $\estwj$ is
unbiased, denote by $\fpaths$ the set of all successful (full) paths
from $G_1$ to $G_n$. Then the sought count is $|\fpaths|$. Indeed,
\[
\mathbb{E}[\estwj]=\sum_{\gamma\in \fpaths}\pr(\gamma)\cdot\estwj(\gamma)=
\sum_{\gamma\in\fpaths}\frac{\pr(\gamma)}{\pr(\gamma)}=|\fpaths|\,.
\]

{\begin{figure}
\begin{tikzpicture}[darkstyle/.style={circle,draw,fill=black,minimum size=0.001}]
  \def\tbls{4} % tablesNum
  \def\insts{5} % instancesNum
  \foreach \x in {1, ..., \tbls} {
    \foreach \y in {1,..., \insts} {
            \filldraw (2*\x-2,-\y) circle (2pt) %node[label=above:$\ifthenelse{\x=1}{a}{\ifthenelse{\x=2}{b}{c}}_\y$] (\x\y) {};
            node[label=above:$t^{\y}_{\x}$] (\x\y) {};
    }
    \node[] at (\x*2-2, -\insts-0.8) {$G_\x$};
  }
    %\node[] at (4.8,-3) {\huge \dots};
    %\draw (5,-3) circle (1pt)
    
    % edge from table #1 index #2 to table #1+1 on index #2
    \def\edg#1#2{\pgfmathtruncatemacro\tbl{#1+1}; \draw (#1#2) -- (\tbl#2);}
    % edge from table #1 index #2 to table #1+1 on index #2+1
    \def\edgn#1#2{\pgfmathtruncatemacro\tbl{#1+1}; \pgfmathtruncatemacro\ind{#2+1} \draw (#1#2) -- (\tbl\ind);}
    % edge from table #1 index #2 to table #1+1 on index #2-1
    \def\edgp#1#2{\pgfmathtruncatemacro\tbl{#1+1}; \pgfmathtruncatemacro\ind{#2-1} \draw (#1#2) -- (\tbl\ind);}

    % edge from table #1 index #2 to table #1+1 on index #2+#3
    \def\edgo#1#2#3{\pgfmathtruncatemacro\tbl{#1+1}; \pgfmathtruncatemacro\ind{#2+#3} \draw (#1#2) -- (\tbl\ind);}
    
    % edge from table #1 index #2 to table #1+1 and all the indexes in the #3 range
    \def\edgr#1#2#3{\pgfmathtruncatemacro\tbl{#1+1}; \foreach \x in #3 { \draw (#1#2) -- (\tbl\x);}}

    \edgr{1}{2}{{1,2,3,4}}    
    \edgr{2}{1}{{1,2,3}} 
    \edgr{2}{2}{{1,2,3,4}}    
    \edgr{2}{3}{{1,2,3,4}} 
    \edgr{3}{2}{{2,3}}

    \edgr{1}{4}{{4,5}}    
    \edgr{2}{4}{{3,4,5}} 
    \edgr{2}{5}{{3,4,5}}    
    \edgr{3}{5}{{4,5}}

%    \edgr{1}{5}{{4,5,6}}    
%    \edgr{2}{4}{{4,5,6}} 
%    \edgr{2}{5}{{4,5,6}}    
%    \edgr{2}{6}{{4,5,6}} 
%    \edgr{3}{5}{{4,5,6}}

%    \edg{1}{10}
%    \edgp{2}{10}

    %node = (\table\index)
    
\end{tikzpicture}
\caption{A join graph\label{fig:general-query-graph}}
\end{figure}
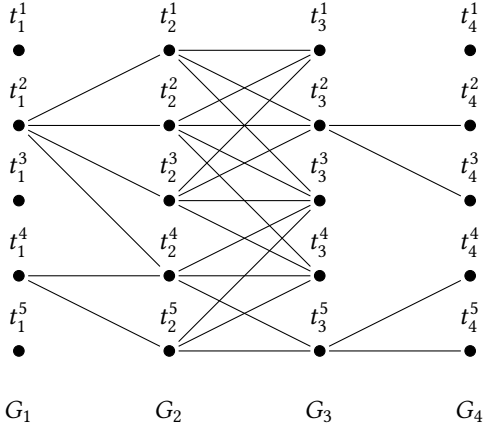}

\begin{example}
  We demonstrate Wander Join using the example join graph depicted in
  Figure~\ref{fig:general-query-graph}. There, each column in the
  figure is a graph $G_i$, and each node $t_{i}^{j}$ is a tuple of
  $G_i$. An edge exists between two tuples if they agree on their join
  attributes. The random walk is from left to right. Choosing the
  random path $\gamma_1=(t_{1}^{2},t_{2}^{2},t_{3}^{2},t_{4}^{2})$
  yields the estimate $\estwj(\gamma_1)=5 \cdot 4 \cdot 4 \cdot
  2=160$. For $\gamma_2=(t_{1}^{4},t_{2}^{5},t_{3}^{5},t_{4}^{5})$
  we will get $\estwj(\gamma_2)=5 \cdot 2 \cdot 3 \cdot
  2=60$. Finally, partial paths, such as
  $(t_{1}^{2},t_{2}^{2},t_{3}^{3})$, will yield the estimate zero. The
  final estimator is the average over all estimates.  \qed
\end{example}

Wander Join adapts to \e{grouping} in a manner similar to that of
Ripple Join~\cite{Haas:1999:RJO:304182.304208}: maintaining a separate
estimator for each group, and using the random $\gamma$ to update only
the separator of the group to which $\gamma$ belongs.

\subsection{\audit}\label{sec:aj}
For presentation sake, we first describe \audit while ignoring the
distinct operator. Again, we also ignore grouping, since \audit is
adapted to grouping similarly to Wander Join and Ripple Join. Hence,
our goal is again to estimate $|\fpaths|$, where
$\fpaths$ is the set of all full random walks $\gamma$ from $G_1$ to
$G_n$. 

Zhao et al.~\cite{Zhao:2018:RSO:3183713.3183739} combine sampling with exact count computation in order to improve the uniformity of sampling join results. We begin by adopting this idea for online aggregation. 

The basic idea is as follows. For a prefix $\delta=(t_1,\dots,t_\ell)$
of a random walk,  denote by $\fpaths_\delta$ the
set of full paths $\gamma$ with the prefix $\delta$. At each step of
the random walk, we make a rough estimation of the complexity of
computing the precise $|A_\delta|$; we describe this estimation in
Section~\ref{sec:tipping}.  If the estimate is low, we actually
compute $|\fpaths_\delta|$ using CTJ (as described in
Section~\ref{sec:ctj}), and then our estimate is
\[\estajn(\delta)\eqdef|\fpaths_\delta|\times\prod_{i=1}^\ell d_i=\frac{|\fpaths_\delta|}{\pr(\delta)}\,.\]
Otherwise, we proceed exactly as Wander Join. In particular, if we
cannot continue in the random walk, or reach a full path, then we use
$\estajn(\delta)=\estwj(\delta)$.

\begin{example}
  We illustrate \audit (without distinct) by continuing our example
  over Figure~\ref{fig:general-query-graph}. Suppose that after the
  random walk $\delta=(t_{1}^{2},t_{2}^{2})$, we choose to run an
  exact evaluation. Then $|\fpaths_\delta|=2$, since there are two
  full paths (ending at $t_{4}^{2}$ and $t_{4}^{3}$) that begin with
  $\delta$. Then the estimation is $\estajn(\delta) =
  |\fpaths_\delta|/\pr(\delta) = 2 \cdot (5 \cdot 4)=40$.  \qed
\end{example}

In the following proposition, we show that $\estajn$ is unbiased by
straightforwardly adapting the argument for Wander Join.

\begin{proposition}\label{prop:non-n-unbiased}
$\estajn$ is an unbiased estimator of 
$|\fpaths|$.
\end{proposition}
\begin{proof}
  Let $\ppaths$ be the set of all paths $\delta$ where \audit decides to
  terminate the path and produce an estimate. This can be because $\delta$ is a
  full path, because it cannot proceed, or because we decide to compute the
  exact $|\fpaths_\delta|$. The reader can verify that, no matter which of
  the three is the case, \audit produces the same estimator, namely
  $|\fpaths_\delta|/\pr(\delta)$. In particular, $|\fpaths_\delta|=1$ in
  the first case and $|\fpaths_\delta|=0$ in the second. We have the
  following.
\begin{align*}
  \mathbb{E}[\estajn]= \sum_{\delta\in
    \ppaths}{\pr(\delta)}\cdot\frac{|\fpaths_\delta|}{\pr(\delta)}= \sum_{\delta\in
    \ppaths}|\fpaths_\delta|=|\fpaths|
\end{align*}
Therefore, $\estajn$ is unbiased, as claimed.
\end{proof}

Note that Audit Join automatically leverages the caching of CTJ,
potentially avoiding re-computation when building the same prefix
$\delta$ in later random walks.

We now extend our estimator to support the distinct operator. Recall
the query of Figure~\ref{fig:general-query}. Our goal is to count the
distinct values taken by $\beta$. For a full path $\gamma$, we denote
by $\beta(\gamma)$ the value to which $\gamma$ assigns $\beta$. Let
$V=\set{\beta(\gamma)\mid \gamma\in \fpaths}$. Our goal is to estimate
$|V|$. For $b\in V$, we denote by $\pr(b)$ the probability that the
random walk reaches a full path $\gamma$ with $\beta(\gamma)=b$; that is,
$\pr(b)$ is the sum of the probabilities of all $\gamma\in \fpaths$
that assign $b$ to $\beta$. Similarly, we denote by $\pr(\delta,b)$ the
probability that the random walk starts with $\delta$ and reaches a
full path $\gamma$ with $\beta(\gamma)=b$; that is, $\pr(\delta,b)$ is the
sum of the probabilities of all $\gamma\in \fpaths$ such that $\alpha$
is a prefix of $\gamma$ and $\gamma$ assigns $b$ to $\beta$. We then combine these probabilities into the following estimator for distinct.

%We use the random partial walks $\delta$ described for
%the non-distinct case. Yet, 

\begin{equation}\label{eq:estaud}
\estajd(\delta)\eqdef
\sum_{b\in V}
\frac%
{\pr(\delta,b)}%
{\pr(\delta)\cdot\pr(b)}
\end{equation}

\begin{example}
  To demonstrate \audit with count distinct, we again use our running
  example over Figure~\ref{fig:general-query-graph}.  For this
  example, suppose that $\beta$ occurs in $G_3$, and moreover, that
  each tuple $t_3^i$ holds a unique value for $\beta$ (while many join
  tuples may include $t_3$).  Suppose that the random walk produces
  $\delta=(t_{1}^{2},t_{2}^{2})$, and that \audit decides to run an
  exact evaluation at this point.  There are two full paths that
  extend $\delta$, both through $(t_{2}^{3})$. We denote by $b$ the
  value of $\beta$ for $(t_{2}^{3})$. From the previous example we get
  that $\pr(\delta) = \frac{1}{20}$. There are three paths leading to
  $t_{2}^{3}$, and by summing their probabilities we get
  $\pr(b)=\frac{1}{5 \cdot 4 \cdot 3}+\frac{2}{5 \cdot 4 \cdot
    4}=\frac{1}{24}$. The last probability of our estimator is
  \[\pr(b,\delta)=\pr(b\mid\delta) \cdot \pr(\delta)=\frac{1}{4} \cdot
  \frac{1}{20}=\frac{1}{80}\,.\] Hence, our estimator yields the
  following.
\[\estajd(\delta) = 
\frac%
{\pr(\delta,b)}%
{\pr(b)\cdot\pr(\delta)} = 
\frac%
{20*24}%
{80}
\]
Hence, the estimate $\estajd(\delta)$ is $6$. \qed
\end{example}

In our implementation, the probability $\pr(b)$ is computed online, after sampling the partial random path $\delta$, by using CTJ to materialize all paths leading to the sampled $b=\beta(\delta)$, summing up their
probabilities, and caching the results. Clearly, this can be an expensive join query, and our
cost estimation (described in Section~\ref{sec:tipping}) guards us
from costly cases.  Nevertheless, as we show in Section~\ref{sec:experiments},
it turns out that in our use case, this computation is very often
tractable after setting $\beta=b$, hence the considerable benefit of \audit.  Next, we show
that the estimator is unbiased.

\begin{proposition}\label{prop:non-d-unbiased}
$\estajd$ is an unbiased estimator of 
$|V|$.
\end{proposition}
\begin{proof}
  Following a similar reasoning as in the proof of
  Proposition~\ref{prop:non-n-unbiased}, we can treat all cases of the
  estimator in a uniform way, that is, according to
  Equation~\eqref{eq:estaud}.  In the following analysis, we identify
  value $b\in V$ with the event that the random walk is complete and,
  moreover, assigns $b$ to $\beta$.  Hence, we have the following.
\begin{align*}
\mathbb{E}[\estajd]=&%
\sum_{\delta\in \ppaths}%
\pr(\delta)\cdot
\sum_{b\in V}
\frac%
{\pr(\delta,b)}%
{\pr(\delta)\cdot\pr(b)}\\
=&
\sum_{\delta\in \ppaths}%
\pr(\delta)\cdot
\frac{1}{\pr(\delta)}
\times
\sum_{b\in V}
{\pr(\delta\mid b)}\\
=&%
\sum_{b\in V}%
\sum_{\delta\in \ppaths}
{\pr(\delta\mid b)}=
\sum_{b\in V}1=%
|V|
\end{align*}
Therefore, $\estajd$ is unbiased, as claimed.
\end{proof}

{
\begin{figure}
\def\assn{\mathrel{{:}{=}}}
\hrule
\vskip0.3em
$\mathsf{AuditJoin}(G_1,\dots,G_n,\alpha,\beta)$
\vskip0.3em
\hrule
\begin{algorithmic}[1]
\State $N\assn 0$
\State $A\assn$ the set of possible assignments for $\alpha$
\State $B\assn$ the set of possible assignments for $\beta$
\Repeat
\State $F_1\assn G_1$
\State $\delta\assn\epsilon$
\For{$i=1,\dots,n$}
\State $N\assn N+1$
\State select $t\in F_i$ randomly and uniformly
\State $\delta\assn (\delta,t)$
\If{$i=n$ or tipping point is reached}
\ForAll{$a\in A$}
\State $C_a \assn C_a+\sum_{b\in B}
\frac%
{\pr(a,b,\delta)}%
{\pr(a,b)\cdot\pr(\delta)}$
\EndFor
\State \textbf{continue} \Comment{Go to line~5}
\EndIf
\State $F_{i+1}\assn G_{i+1}\ltimes t$
\If{$F_{i+1}=\emptyset$}
\State \textbf{continue} \Comment{Go to line~5}
\EndIf
\EndFor
\Until{time limit is reached}
\ForAll{$a\in A$}
\State estimate count-distinct for $a$ as $C_a/N$
\EndFor
\end{algorithmic}
\hrule
\caption{\label{fig:aj}Audit Join pseudo code}
\end{figure}}

\vspace{-1em}

\subsubsection{Tipping Point}\label{sec:tipping}
To decide when to use partial exact computations, we use a rough
estimate of the join complexity. We do so using a simple
technique for join-size estimation as used by
PostgreSQL.\footnote{\url{https://www.postgresql.org/docs/current/static/planner-stats-details.html}}
In the case of two triple patterns $(a_1,b_1,c_1)$ and $(a_2,b_2,c_2)$ joining on say $c_1 = c_2$, the size is estimated as the product between the number of triples matched by $(a_1,b_1,c_1)$ and $(a_2,b_2,c_2)$, divided by the maximum number of distinct terms  of $c_1$ or $c_2$. For more than two patterns, we compose the estimates in the straightforward manner. If the estimate is lower than
a predefined threshold, \audit switches to exact computation.
In this case we say that the \e{tipping point} is reached.
While simple, this join estimation allows \audit to consistently achieve considerable improvements, as shown in the experimental section.
Investigating more sophisticated estimates
(e.g.,~\cite{DBLP:journals/pvldb/VengerovMZC15}) is left as an important direction for future research.

\subsubsection{Summary}
We summarize  \audit in Figure~\ref{fig:aj}. The code is similar
to  Section~\ref{sec:aj}, except that it
incorporates grouping. The sets $A$ and $B$ are projections over 
the attributes $\alpha$ and $\beta$, respectively (Figure~\ref{fig:general-query}).
The estimates are accumulated in $C_a$ for
every group $a$. The probability ${\pr(a,b,\delta)}$ corresponds to
the event that the random walk starts with $\delta$ and includes the
group $a$ and the counted value $b$.  Hence, it is the sum of the
probabilities of all such random walks. Similarly, ${\pr(a,b)}$ is the
probability that the random walk includes the group $a$ and the
counted value $b$. Note that in line~17, the left semi-join
$G_{i+1}\ltimes t$ consists of all the tuples of
$G_{i+1}$ that can be matched with $t$.

	\section{Experimental Study} \label{sec:experiments}

Our experimental study compares the performance of our four query engine strategies: Virtuoso, Cached Trie Join (CTJ), Wander Join (WJ) and \audit (AJ). More specifically, in the context of answering a variety of queries in our exploration system, we address the following core questions:

\begin{itemize}
    \item[Q1:] How does the performance of the two exact computation strategies -- Virtuoso and CTJ -- compare in the distinct case required by our system?
    \item[Q2:] How does the error rate of the two online aggregation strategies -- WJ and AJ -- compare over time in both the distinct and non-distinct case?
    %\item[Q3:] What error is present in the online aggregation strategies in the time taken by exact computation strategies?
\end{itemize}

\begin{figure*}[ht]
        \centering
        %\vspace*{-2ex}
        \begin{subfigure}[t]{0.3\textwidth}
            \includegraphics[width=\textwidth, trim={0 0 0 1.32cm}, clip=true]{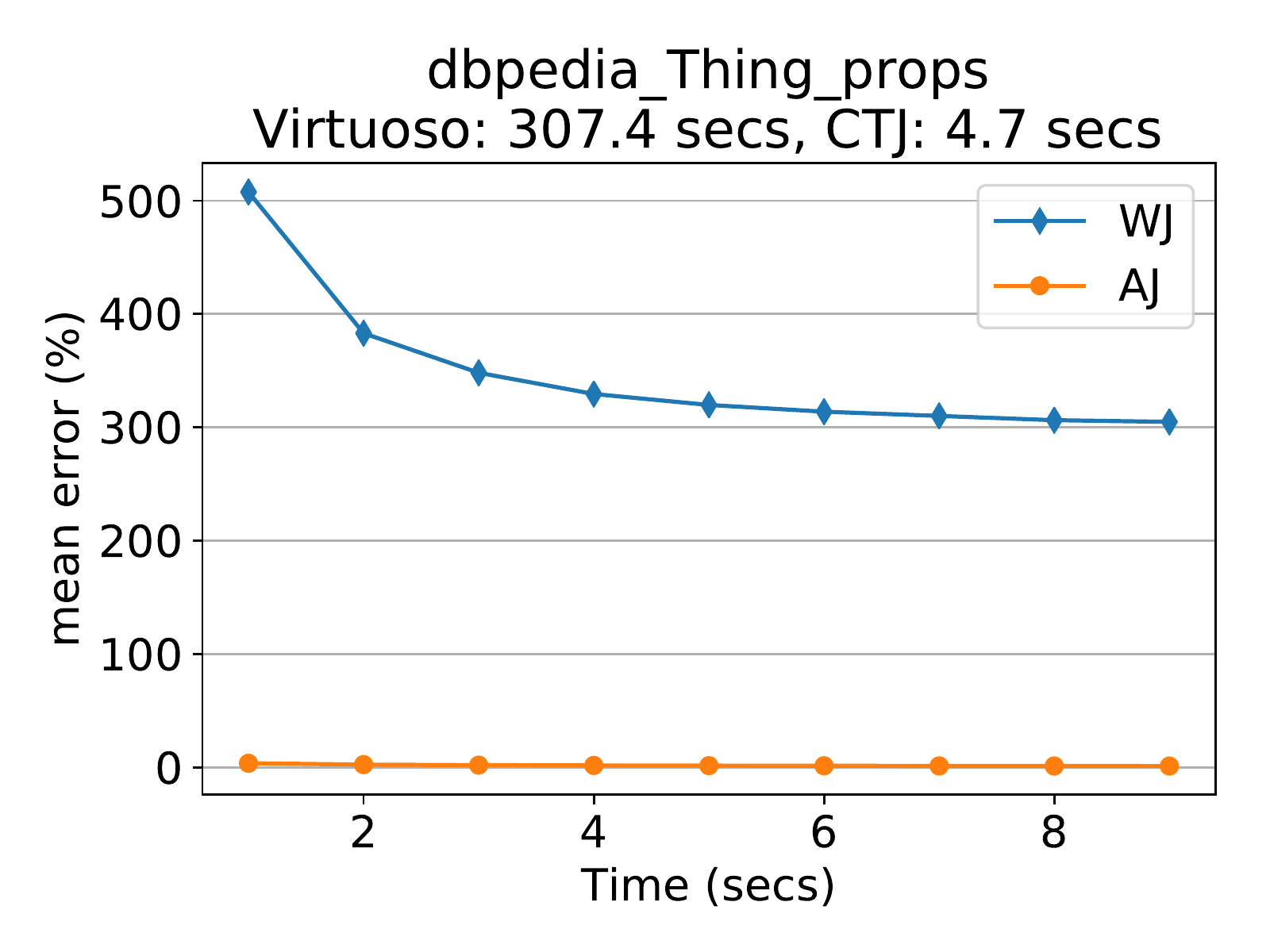}
            \caption{Out-property expansion of\\ class \texttt{Thing} \label{fig:systemd1}}
        \end{subfigure}
        %\hspace*{-2ex}
        \begin{subfigure}[t]{0.3\textwidth}
            \includegraphics[width=\textwidth, trim={0 0 0 1.32cm},clip=true]{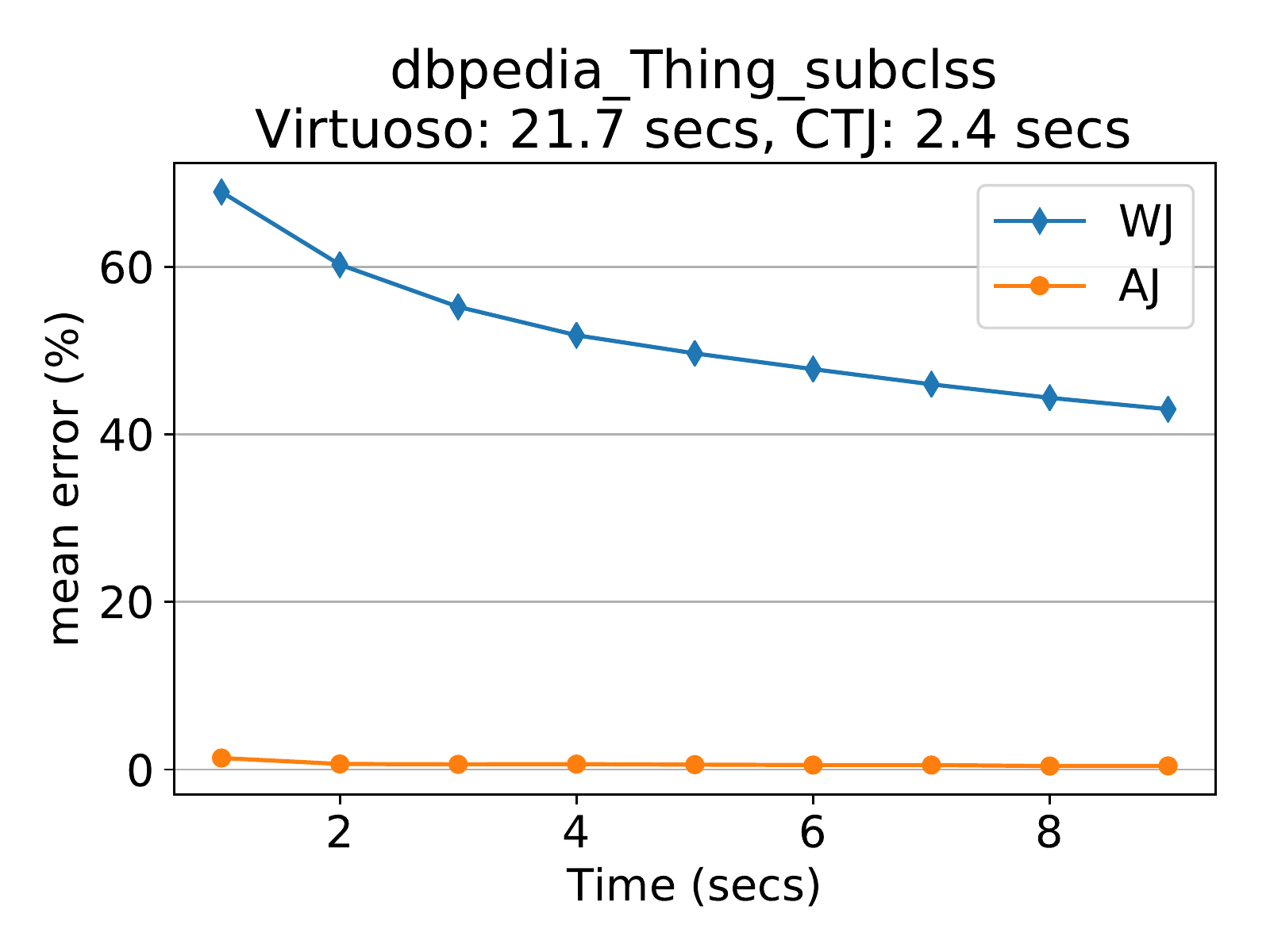}
            \caption{Subclass expansion of\\ class \texttt{Thing} \label{fig:systemd2}}
        \end{subfigure}
        \begin{subfigure}[t]{0.3\textwidth}
            \includegraphics[width=\textwidth, trim={0 0 0 1.32cm},clip=true]{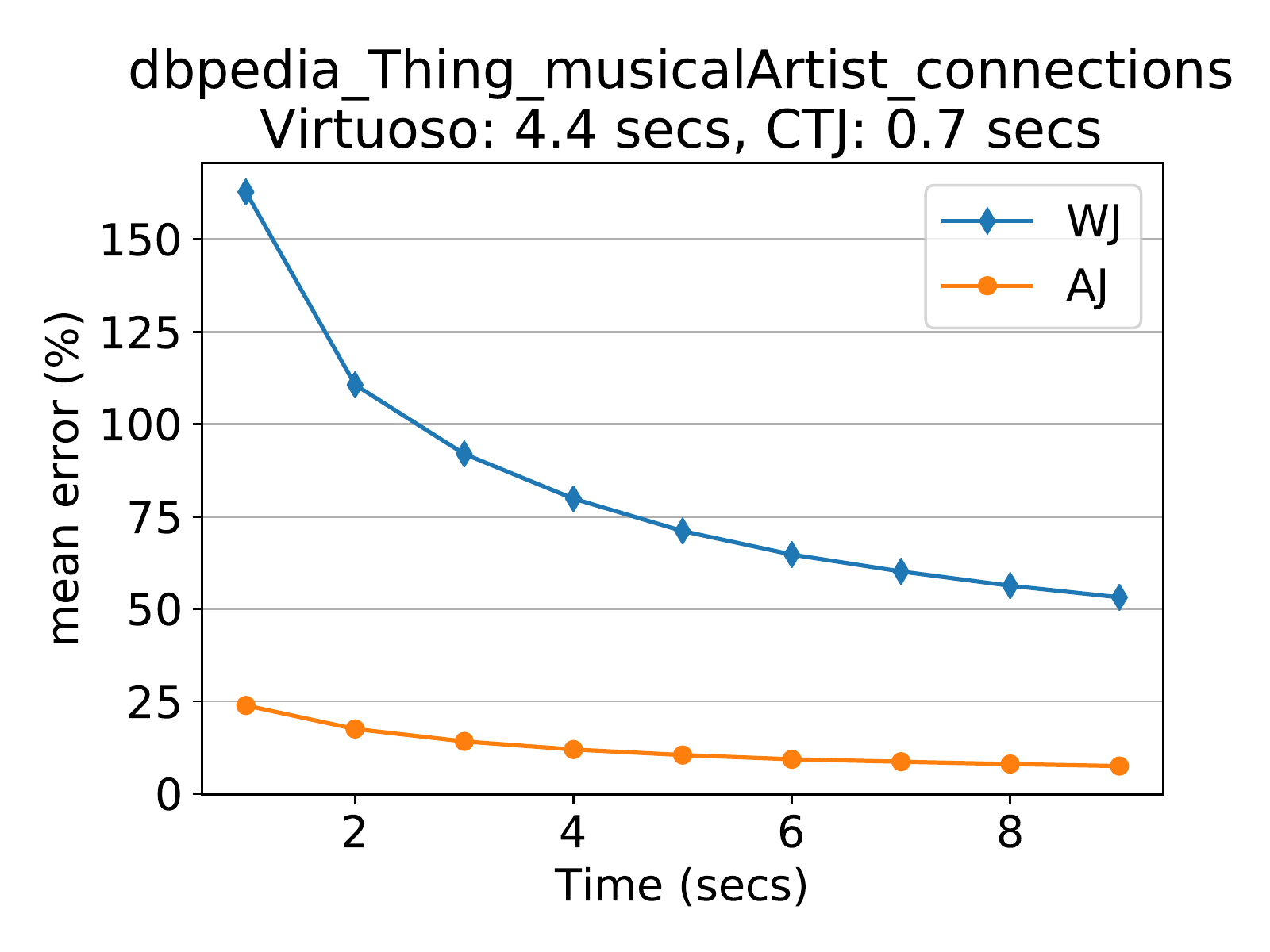}
            \caption{Object expansion of\\ property \texttt{musicalArtist} \\(subjects of type \texttt{Thing}) \label{fig:systemd3}}
        \end{subfigure}
        
        \par\bigskip
         \vspace*{-0.4ex}
        
        \begin{subfigure}[t]{0.3\textwidth}
            \includegraphics[width=\textwidth, trim={0 0 0 1.32cm},clip=true]{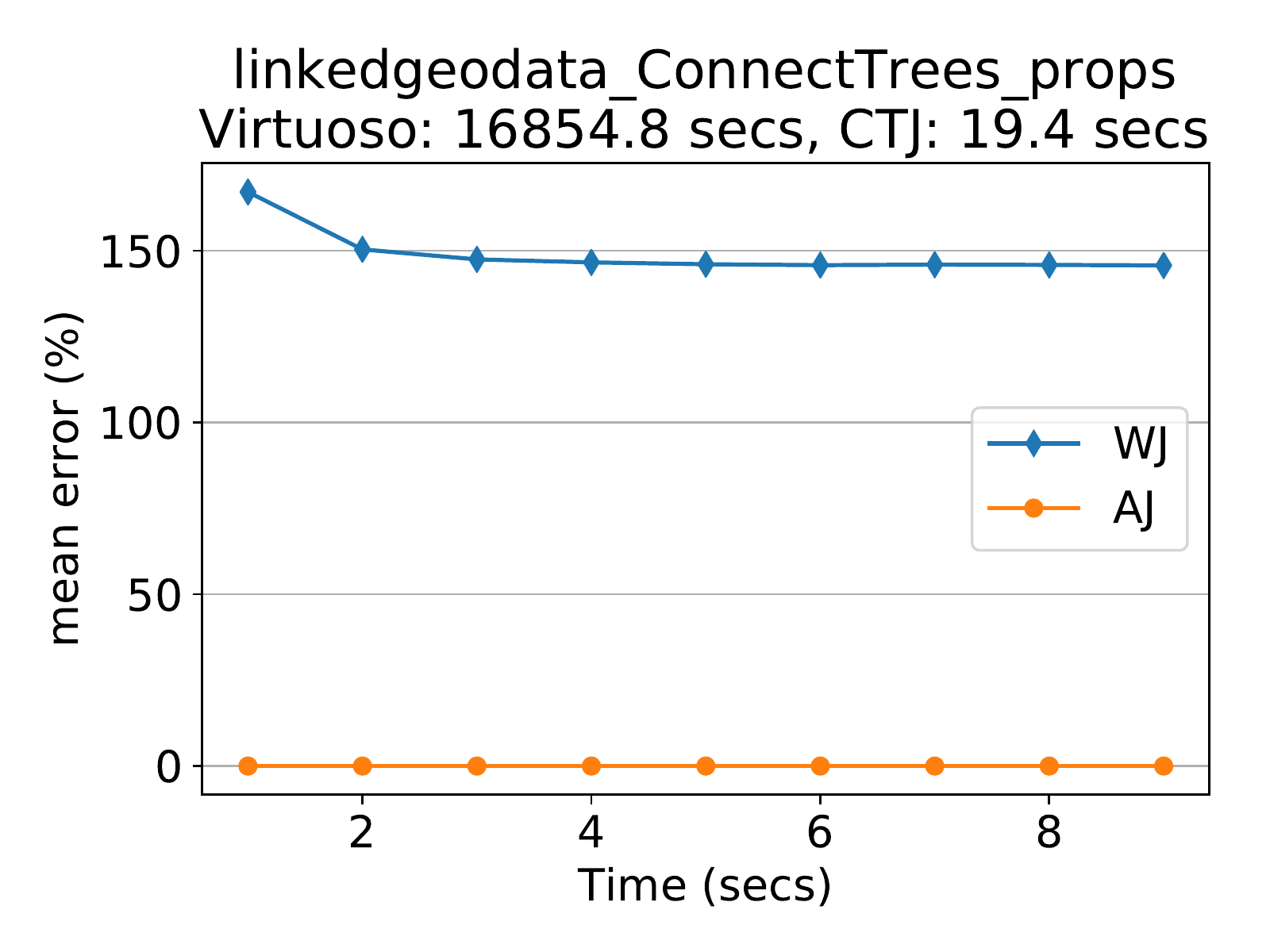}
            \vskip-1em
            \caption{Out-property expansion of\\ class \texttt{Thing} \label{fig:systeml1}}
        \end{subfigure}
        \begin{subfigure}[t]{0.3\textwidth}
            \includegraphics[width=\textwidth, trim={0 0 0 1.32cm},clip=true]{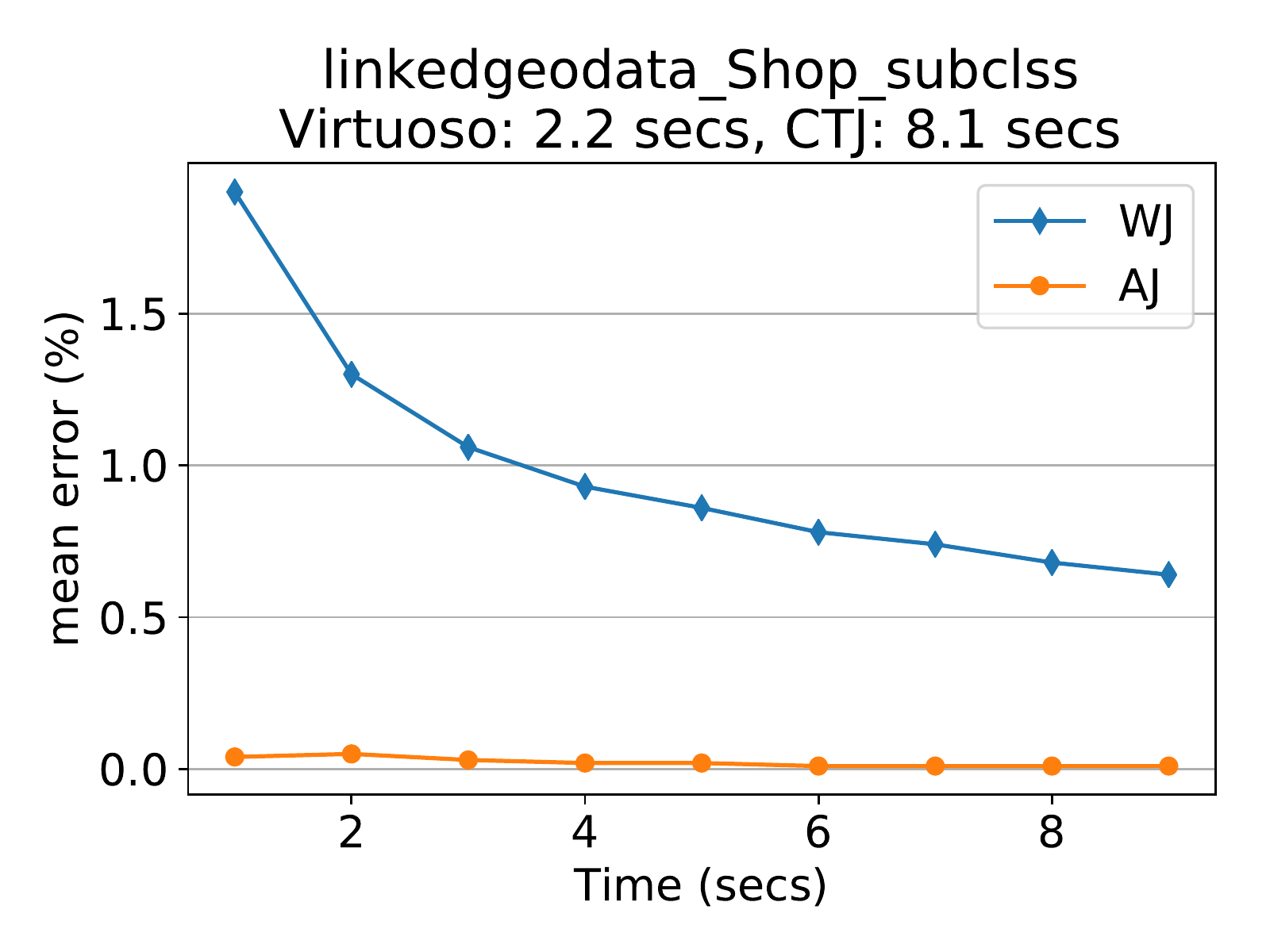}
            \vskip-1em
            \caption{Subclass expansion of\\ class \texttt{Shop} \label{fig:systeml2}}
        \end{subfigure}
        \begin{subfigure}[t]{0.3\textwidth}
            \includegraphics[width=\textwidth, trim={0 0 0 1.32cm},clip=true]{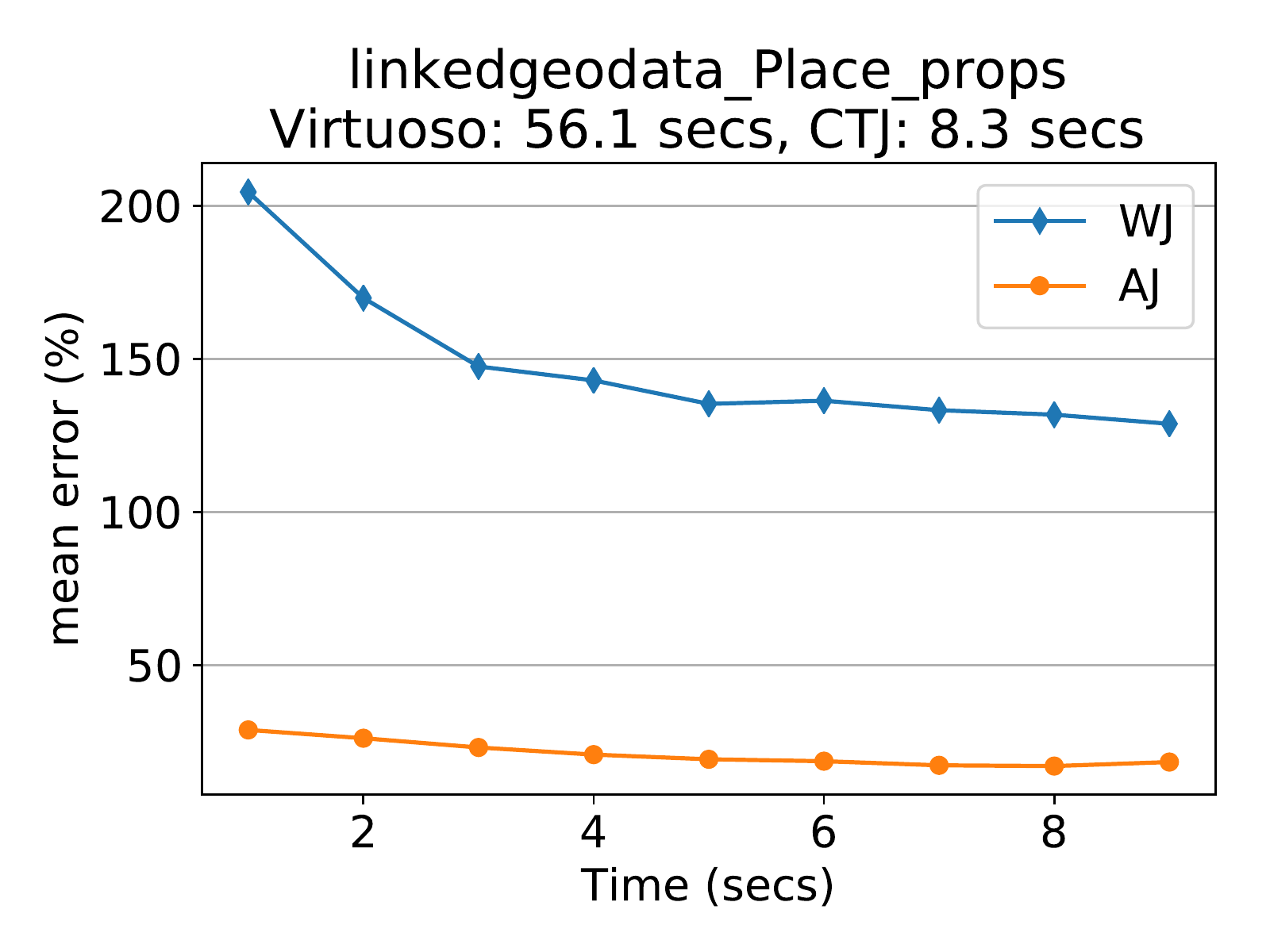}
            \vskip-1em
            \caption{Out-property expansion of\\ class \texttt{Place} \label{fig:systeml3}}
        \end{subfigure}
        \vspace*{-1ex}
        \caption{The mean error per second (WJ and AJ) and exact runtimes (Virtuoso and CTJ) for a selection of six exploration queries: three from DBpedia (top) and three from LinkedGeoData (bottom) \label{fig:systems}}
\end{figure*}

\subsection{Implementation}

We take Virtuoso as an off-the-shelf query engine and implement the other three strategies in C++. We implement WJ and CTJ as described in their corresponding papers and then implement AJ on top of these algorithms. We now discuss the configuration and implementation of each strategy.

We use Virtuoso v.~07.20.3217~\cite{Virtuoso} and configure it to use all available memory for caching indexes. Virtuoso runs the subclass closure as part of the expansion queries using property paths on the original graph; the query planner shows that this closure is executed first and takes only a few milliseconds for both knowledge bases. 

Our CTJ implementation uses LFTJ~\cite{DBLP:conf/icdt/Veldhuizen14} as the trie join algorithm. Since WJ does not support transitivity, the subclass closure is computed offline and materialized in the graph (instances are indexed with only their explicit types per the original data). The trie indexes are implemented using sorted arrays (\e{std::vector}) such that each search is done in $O(log(N))$. We implement CTJ with four of the six possible index orders: $(s,p,o)$, $(o,p,s)$, $(p,s,o)$, and $(p,o,s)$; these orders are sufficient to support our exploration queries. The CTJ cache structure uses an array of hashtables (\e{std::unordered\_map}[]).

In WJ, the graph is saved in an unsorted array (\e{std::vector}). Analogously to CTJ, the subclass closure is performed offline and materialized. The algorithm uses hashtable indexes (\e{std::unordered\_map}) over the array that enable sampling in $O(1)$ time.  In the case of distinct -- as there is no formal support for this operator in WJ -- we augment it with the technique proposed by Haas et al. in Ripple Join~\cite{Haas:1999:RJO:304182.304208, Hellerstein:1997:OA:253260.253291} for performing online aggregation in the distinct case: this technique stores the set of samples seen thus far and rejects new samples that have already been seen.

AJ is implemented on top of WJ and CTJ; it likewise assumes that the subclass closure has been materialized in the graph. Our implementation uses a hybrid hashtable/trie data-structure where the hashtable indexes point to a sorted array, allowing $O(1)$-time sampling for WJ and $O(log(n))$-time search for CTJ. Analogously to CTJ, AJ maintains four index orders and the same caching structure. Unlike WJ, AJ uses its own unbiased estimator for distinct (per Section~\ref{sec:engine}).

\subsection{Methodology}

\paragraph{Data}
Our data include two large-scale knowledge graphs: DBpedia~\cite{Bizer:2009:DCP:1640541.1640848} and LinkedGeoData~\cite{10.1007/978-3-642-04930-9_46}; details of these two graphs are presented in Table~\ref{tbl:datasets}. DBpedia contains multi-domain data extracted from Wikimedia projects such as Wikipedia; we take the English version of DBpedia~v.3.6, which contains 400 million RDF triples, 370 thousand classes and 62 thousand properties. LinkedGeoData specializes in spatial data extracted from OpenStreetMap; we take the November 2015 version, which contains 1.2 billion RDF triples, one thousand classes, and 33 thousand properties. In the case of LinkedGeoData, since no root class is defined in the original data, we explicitly add a class that is the parent of all classes previously without a parent in the class hierarchy.

\begin{table}
    \caption{Dataset information \label{tbl:datasets}}
    \vspace{-0.7em}
    \small
    \begin{tabular}{ l l r r r}
    \toprule
    \textbf{Dataset} & \textbf{Version} & \textbf{Triples} & \textbf{Classes} & \textbf{Props} \\ \midrule
    DBpedia & 3.6 & 431,940,462 & 370,082 & 61,944 \\ 
    LinkedGeoData & 2015-11 & 1,216,585,356 & 1,147 & 33,355 \\
    \bottomrule
    \end{tabular}
    %\vskip-1em
\end{table}

\paragraph{Queries} Our queries comprise of randomly generated exploration paths that imitate users applying incremental expansions. More specifically, our generator starts with the root class of a graph. At each step, the generator uniformly selects one of the expansion operations, which is translated to a SPARQL query of the form shown in Figure~\ref{fig:general-query}. Next, one of the groups (aka. bar) from the answer is randomly sampled; we apply a weighted sampling according to the size of the group to increase focus on large groups (otherwise since the majority of groups are small, the explorations would fixate on these small groups). The generator continues for four steps or until it gets an empty result. Queries with empty results are ignored and not considered part of the path. We ran this generator 25 times for each graph resulting in a total of 33 distinct non-empty queries (listed in the appendix). The error is computed as the absolute difference between the exact count and estimated count divided by the exact result; consequently, the reported mean error is the average error over all groups in the result.

\paragraph{Machine and Testing Protocol}
Our server has four 2.1~GHz Xeon E5-2683 v4 processors, 500~GB of DDR4 DRAM, and runs Ubuntu~16.04.4 Linux. Each experiment was performed three times, and the average runtime and mean error are reported. We run each online aggregation algorithm for nine seconds and report the estimation after each second. For each query, we tested different join orders of WJ and selected the one with the best mean error.

\begin{figure*}[t]
        \centering
        %\vspace*{-2ex}
        \begin{subfigure}[b]{0.245\textwidth}
            \includegraphics[width=\textwidth,clip=true]{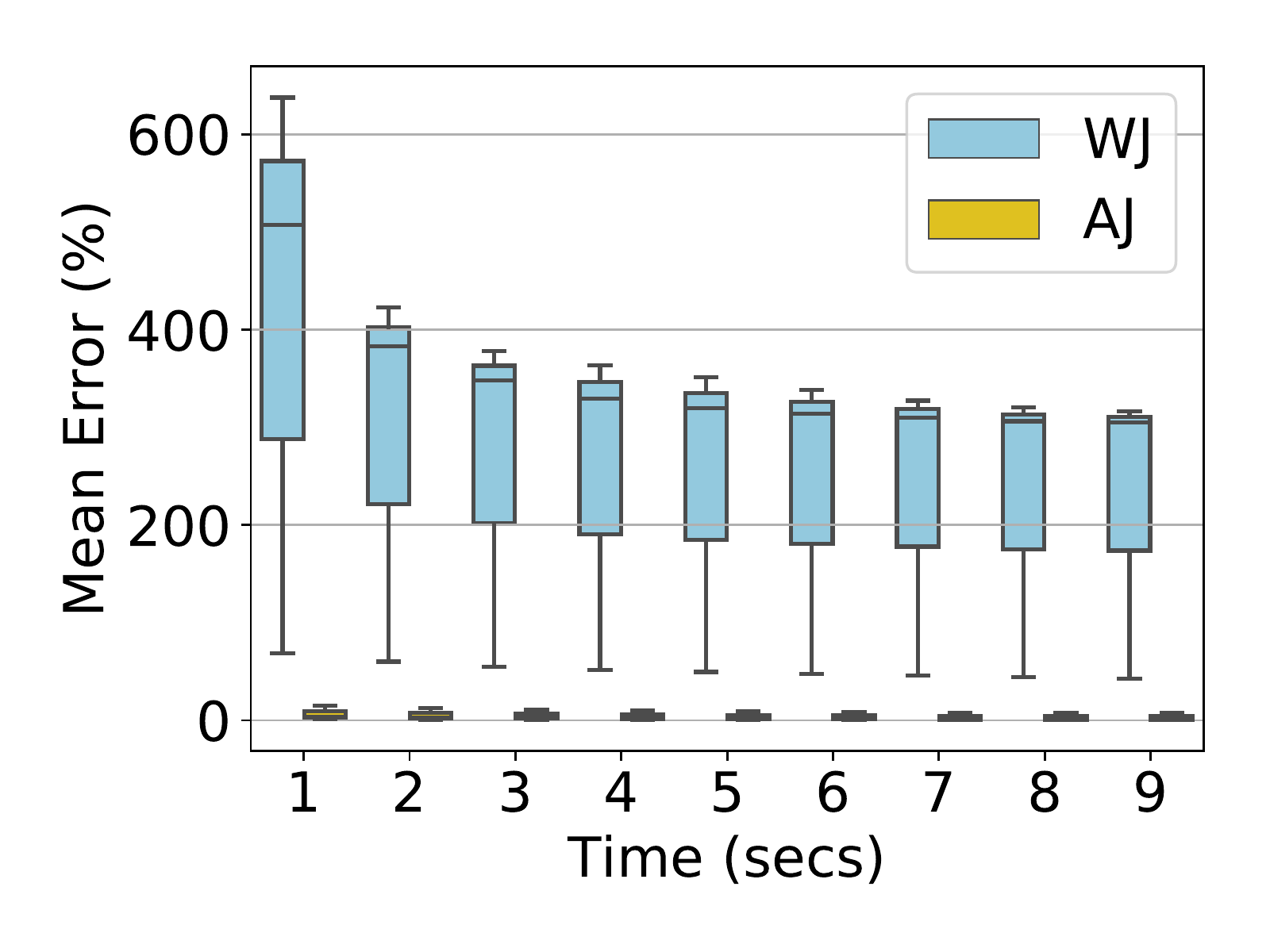}\vskip-1ex
            \caption{Step one queries on DBpedia \label{fig:whisker1d}}
        \end{subfigure}
        \begin{subfigure}[b]{0.245\textwidth}
            \includegraphics[width=\textwidth,clip=true]{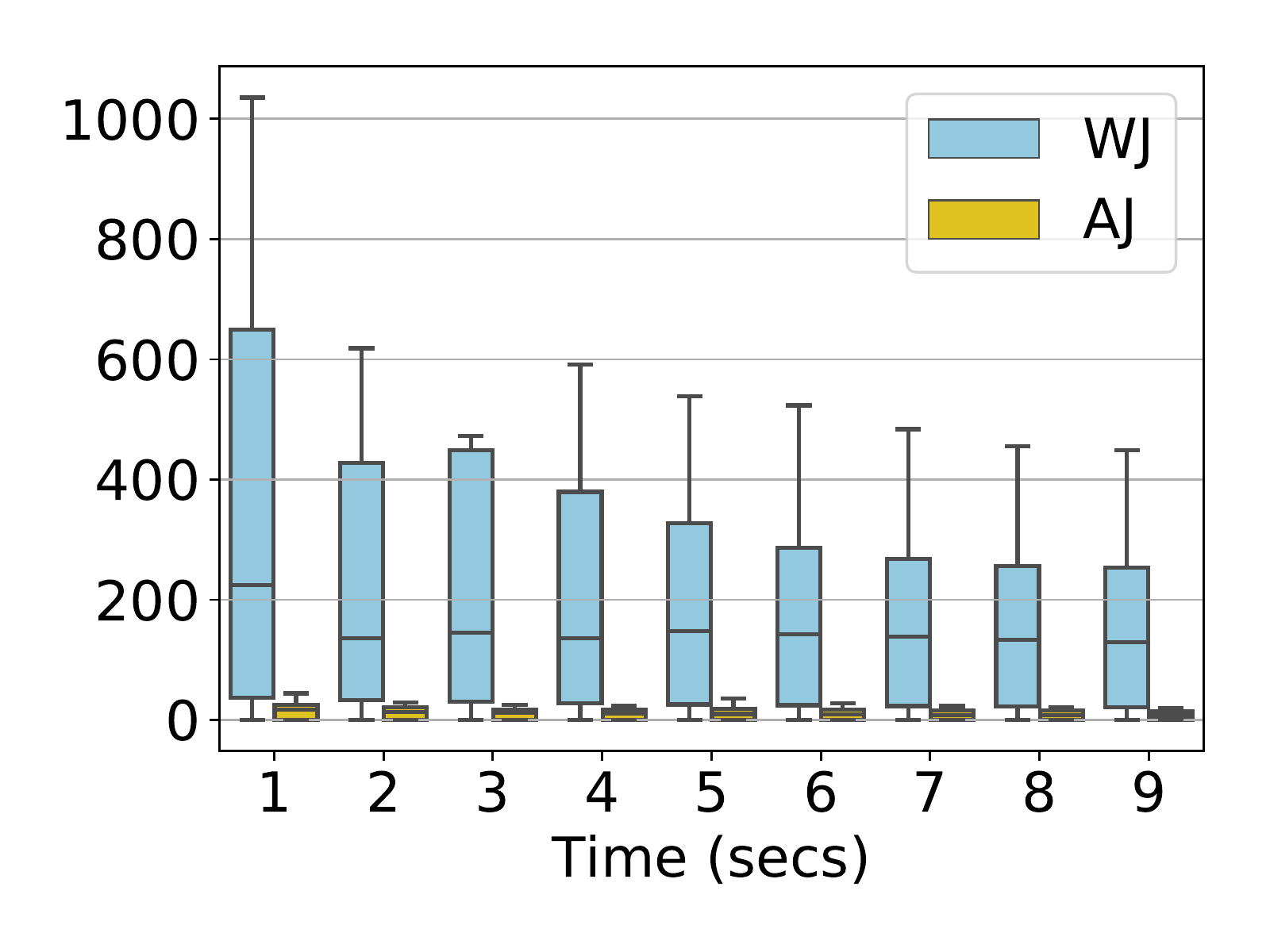}
            \vskip-1ex
            \caption{Step two queries on DBpedia \label{fig:whisker2d}}
        \end{subfigure}
        \begin{subfigure}[b]{0.245\textwidth}
            \includegraphics[width=\textwidth,clip=true]{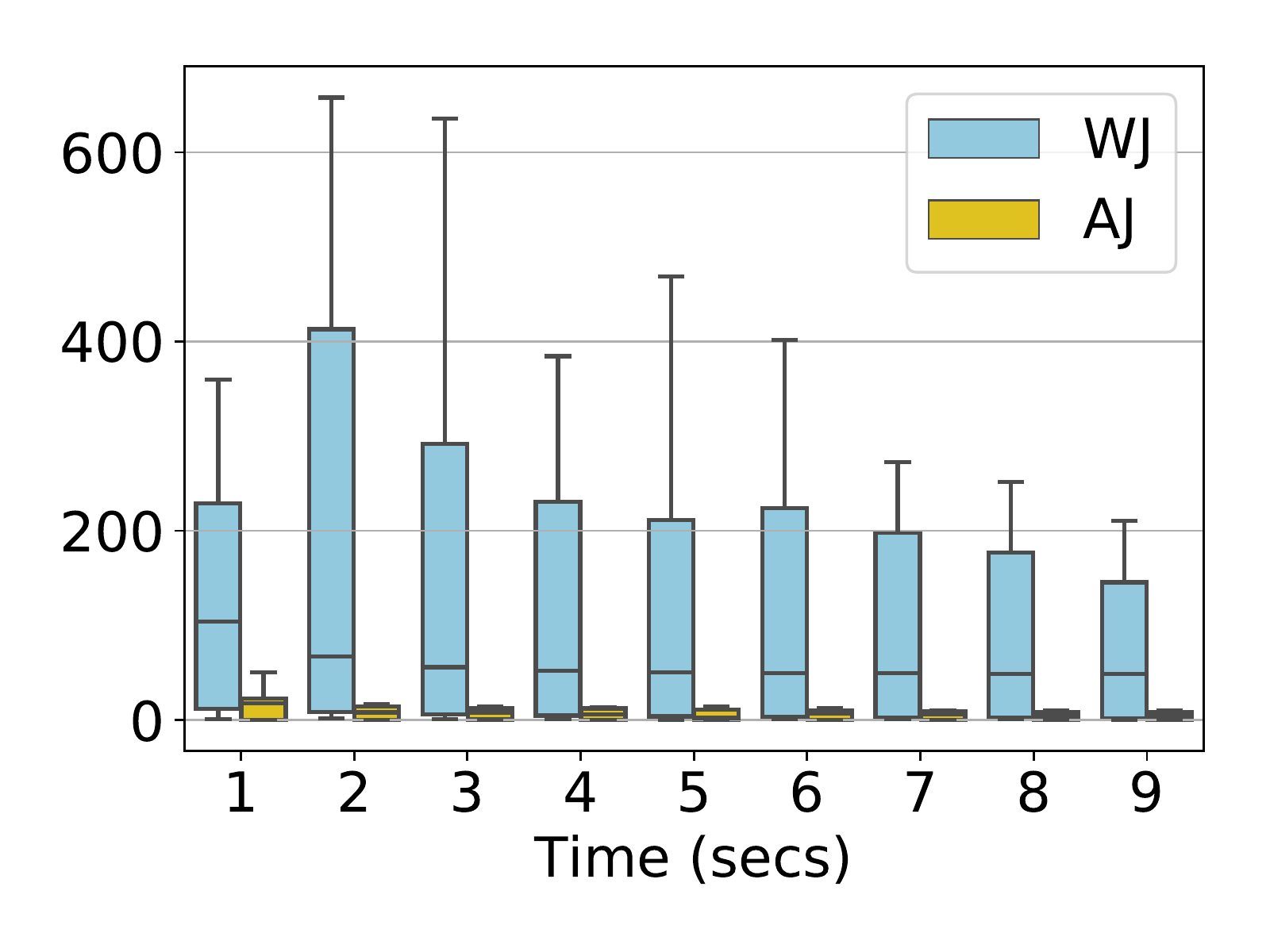}
            \vskip-1ex
            \caption{Step three queries on DBpedia \label{fig:whisker3d}}
        \end{subfigure}
        \begin{subfigure}[b]{0.245\textwidth}
            \includegraphics[width=\textwidth,clip=true]{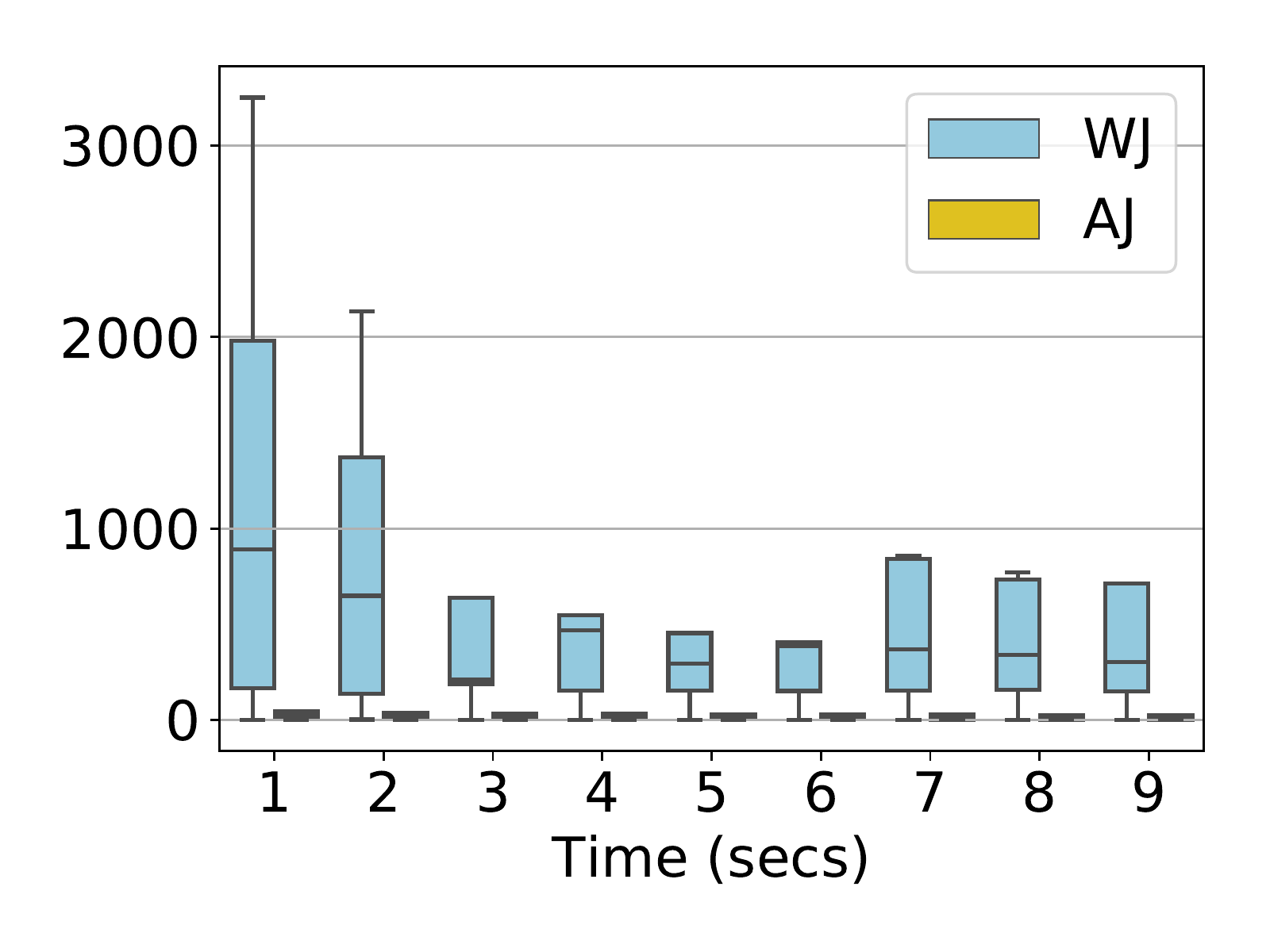}
            \vskip-1ex
            \caption{Step four queries on DBpedia \label{fig:whisker4d}}
        \end{subfigure}

        \par\vskip0.2em

        \begin{subfigure}[b]{0.245\textwidth}
            \includegraphics[width=\textwidth,clip=true]{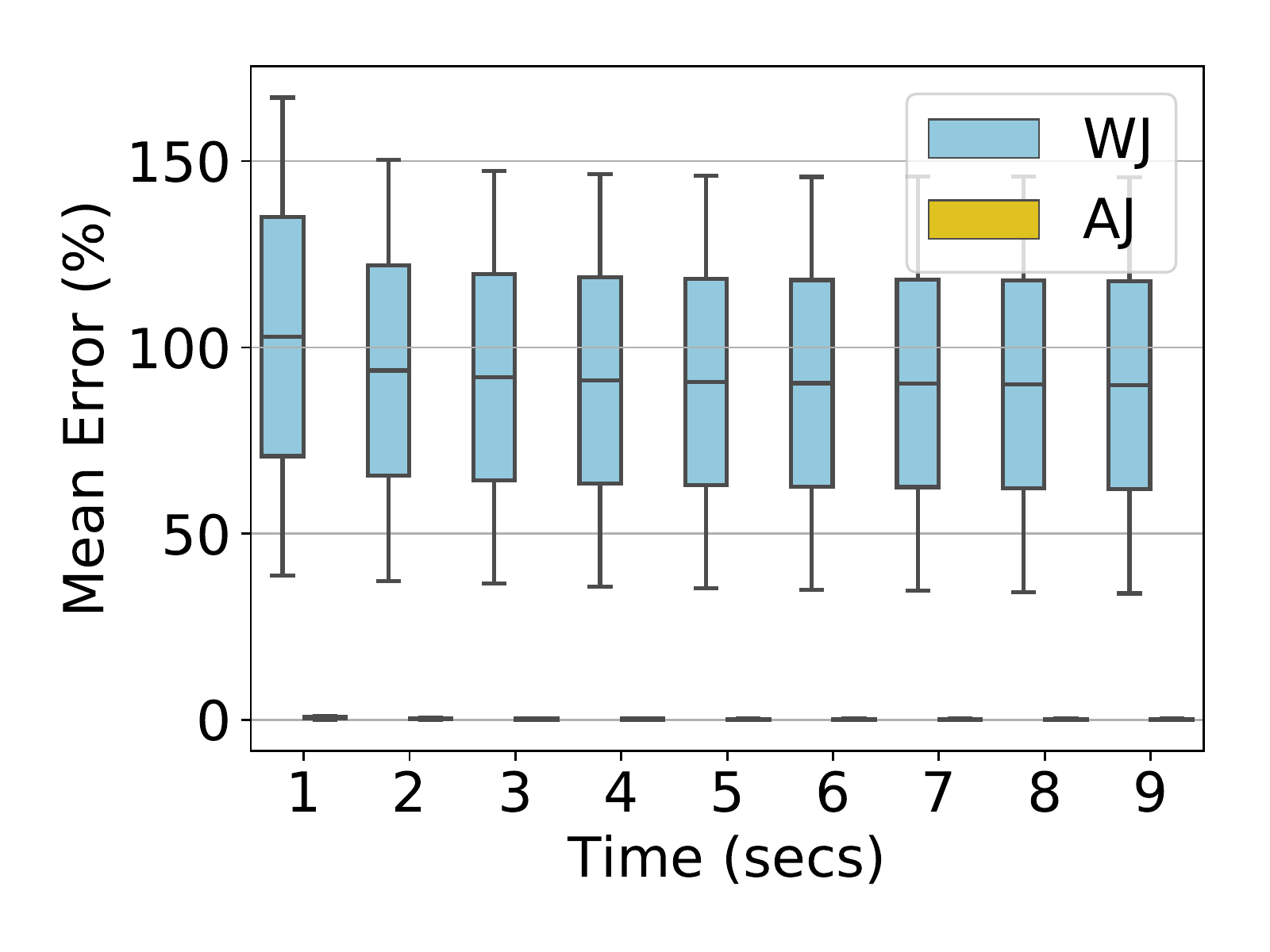}\vspace{-1ex}
            \caption{Step one queries on\\ LinkedGeoData \label{fig:whisker1l}}
        \end{subfigure}
        \begin{subfigure}[b]{0.245\textwidth}
            \includegraphics[width=\textwidth,clip=true]{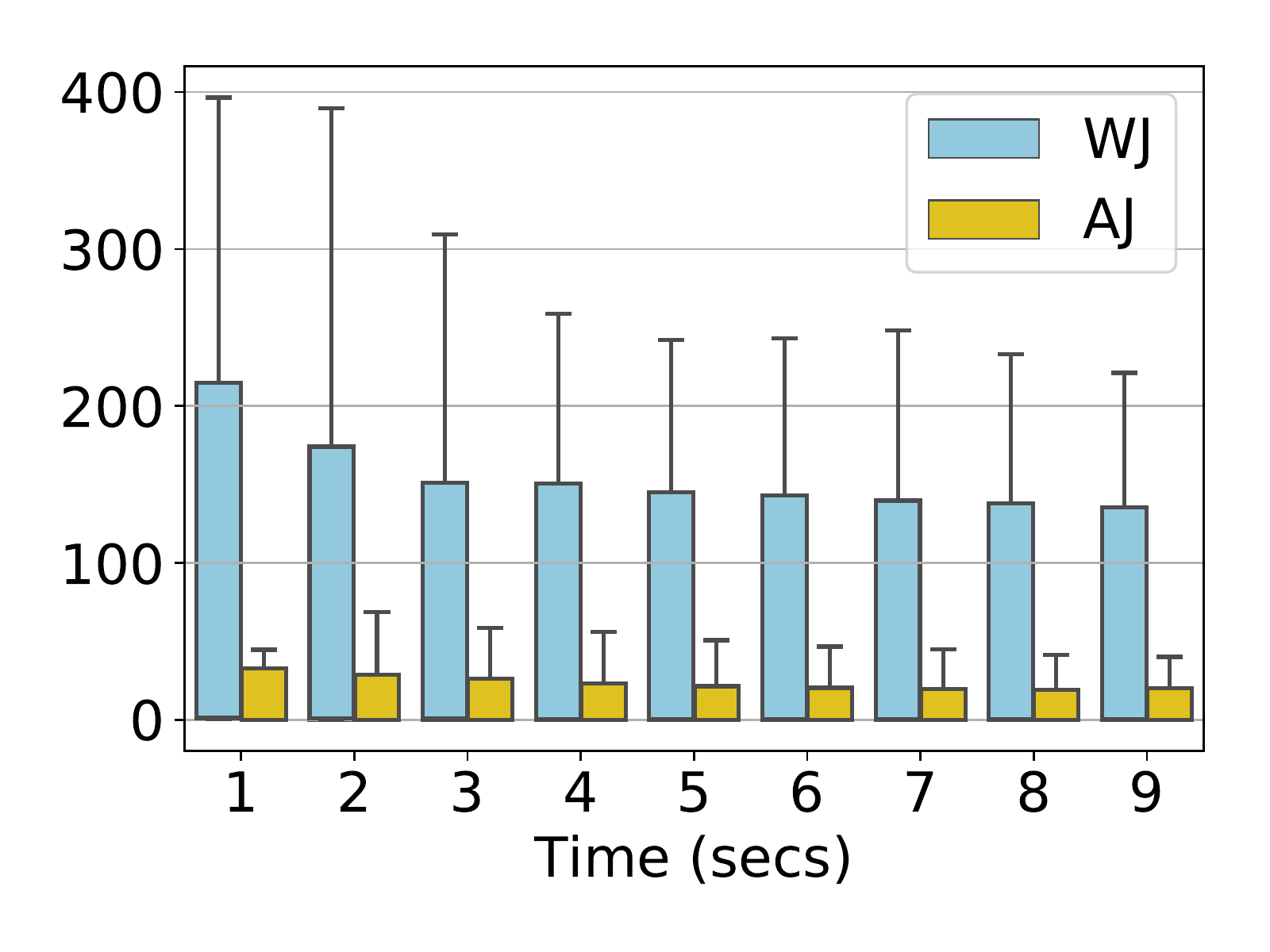}\vspace{-1ex}
            \caption{Step two queries on\\ LinkedGeoData \label{fig:whisker2l}}
        \end{subfigure}
        \begin{subfigure}[b]{0.245\textwidth}
            \includegraphics[width=\textwidth,clip=true]{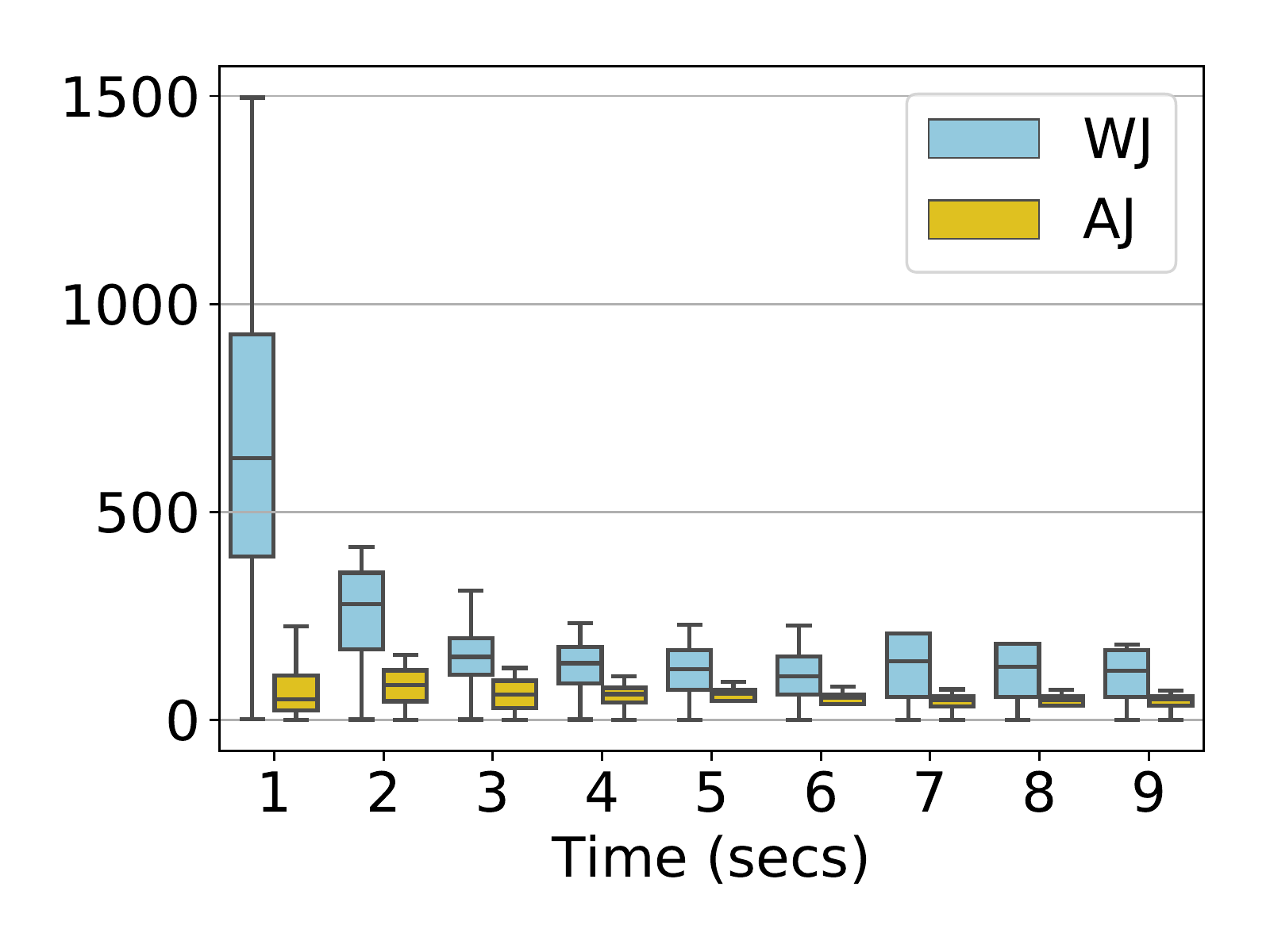}\vspace{-1ex}
            \caption{Step three queries on\\ LinkedGeoData \label{fig:whisker3l}}
        \end{subfigure}
        \begin{subfigure}[b]{0.245\textwidth}
            \includegraphics[width=\textwidth,clip=true]{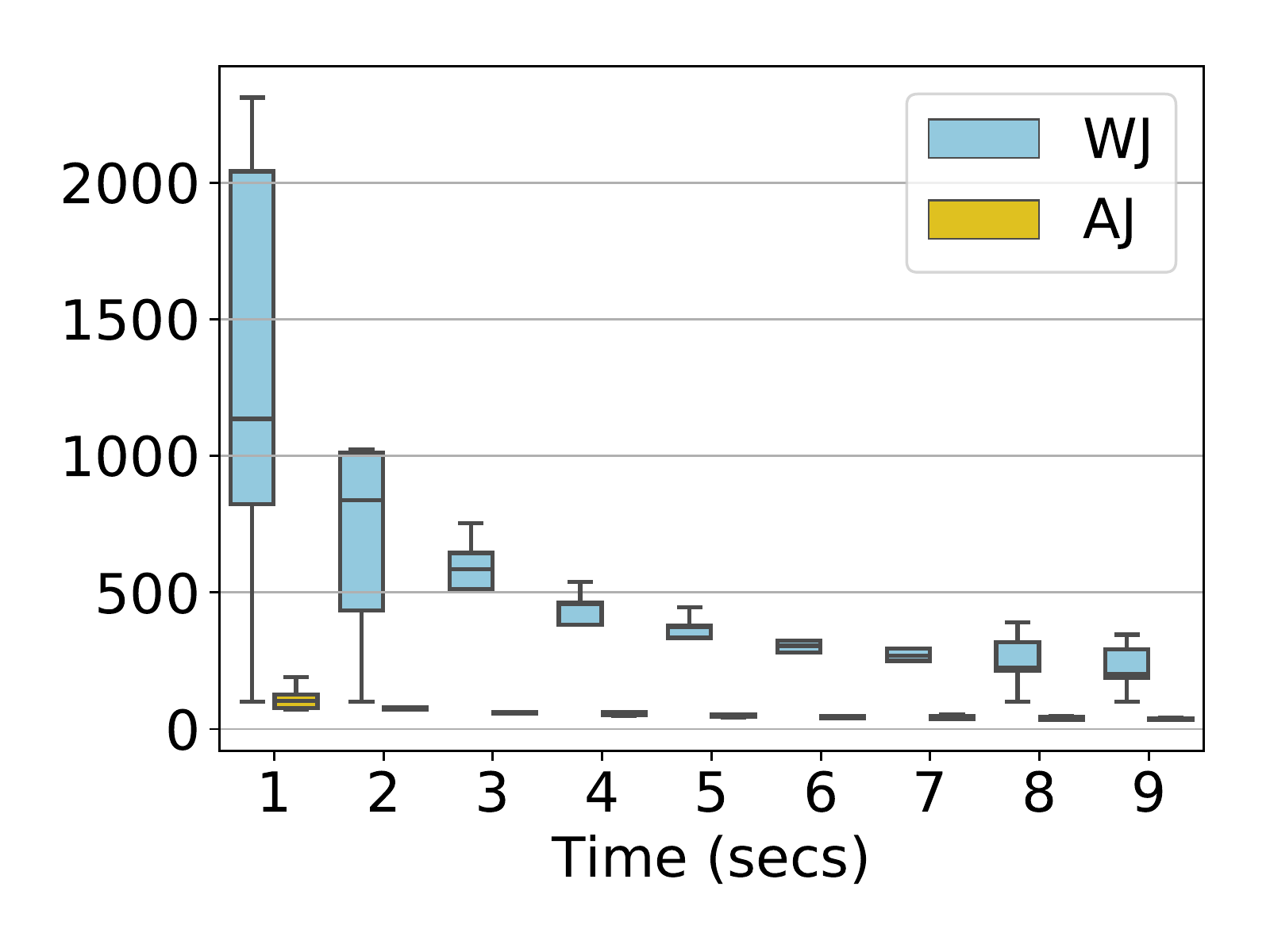}\vspace{-1ex}
            \caption{Step four queries on\\ LinkedGeoData \label{fig:whisker4l}}
        \end{subfigure}

        \vspace*{-1ex}
        \caption{Tukey plots for the mean error over time (seconds) of all queries with a varying number of exploration steps over DBpedia (top) and LinkedGeoData (bottom) \label{fig:whisker}}
\end{figure*}

\subsection{Results}

We first present results on a selection of six queries that help demonstrate different behaviour in all four compared approaches. We then compare the error observed over time for the two online aggregation approaches over all queries, contrasting different exploration depths, the two different datasets, and queries with and without distinct.

\paragraph{Selected queries with distinct:} Figure~\ref{fig:systems} presents the results for a selection of six queries (with distinct) that illustrate a variety of behaviours in the compared approaches. Each graph presents the mean error over time, in seconds, for a specific exploration query. The times for the exact engines, specifically Virtuoso and CTJ, are indicated above the graph. The top row presents results over DBpedia, while the bottom row presents results over LinkedGeoData.

% discuss the offline engines and their irrelevance for the system
Virtuoso and CTJ take more than a second to answer most queries and take notably longer for the larger dataset: LinkedGeoData (which contains three times more edges). Focusing on Virtuoso, we see that while some queries run in the order of seconds, others run in minutes or even hours: the out-property expansion of \texttt{Thing} (the root class) on LinkedGeoData runs for more than four hours (see  Figure~\ref{fig:systeml1}). On the other hand, CTJ generally offers a major performance improvement over Virtuoso: though slower in one case (see Figure~\ref{fig:systeml2}), in the worst cases, CTJ returns results in tens of seconds; for example, the query that took Virtuoso over 4 hours takes CTJ around 20 seconds. Still, even with the considerable performance improvements offered by CTJ, runtimes in the order of tens of seconds would hurt the interactivity and usability of our exploration system.

%For example, CTJ computes the DBpedia \rdftext{Thing} subclass expansion in 2.4 seconds as presented in Figure~\ref{fig:systemd1}, and the LinkedGeoData \rdftext{Thing} Out-Property expansion in 19.4 seconds as presented in Figure~\ref{fig:systeml2}. 
%As mentioned before, this behaviour hurts the interactivity and usability of our system. Therefore, we will focus on the online aggregation engines for the remainder of this section.

% discuss WJ performance compared to AJ
When comparing the mean errors of WJ and AJ over time across Figure~\ref{fig:systems}, it is clear that the accuracy and convergence of AJ is significantly better than WJ for the selected queries. These improvements are due to the two extensions that AJ makes over WJ: namely the reduction of rejection rates using CTJ, and the addition of an unbiased estimator for the distinct case (we shall test without distinct in later experiments). We now discuss some individual cases in more detail.

% separate graph analysis DBP
Looking first at DBpedia, for the out-property expansion of \texttt{Thing} (Figure~\ref{fig:systemd1}), the mean error of WJ  is 519\% after one second and 303\% after nine seconds; the corresponding errors for AJ are 7.5\% and 3.7\%, respectively---almost two orders of magnitude improvement compared to WJ. In the object expansion of \texttt{musicalArtist} (Figure~\ref{fig:systemd3}), which originates from the subclass expansion of \texttt{Thing} (Figure~\ref{fig:systemd2}), WJ has a mean error of 163\% and 53\% after one and nine seconds, respectively. While still better than WJ, the accuracy of AJ drops, starting at a mean error of 24\% and reaching 9\% (on the other hand, given the high selectivity of the query, CTJ returns exact results in less than a second).

% separate graph analysis LGD
Moving to the second row of plots in Figure~\ref{fig:systems}, while the larger size of LinkedGeoData notably affects Virtuoso and CTJ, the results for the online aggregation solutions remain similar to DBpedia. As shown in Figure~\ref{fig:systeml1}, WJ's estimation of the out-property expansion on \texttt{Thing} results in a mean error of 167\% after one second, which slowly drops to 145\% after nine seconds; for the same query, AJ estimates the results with a mean error close to 0\% from the first second. The subclass expansion on \texttt{Shop} (Figure~\ref{fig:systeml2}) is quickly well-estimated by both WJ and AJ, resulting in a mean error of 1.9\% and 0.04\% after one second, respectively. On the other hand, estimations worsen in the third query for an out-property expansion of \texttt{Place} (Figure~\ref{fig:systeml3}); still, AJ clearly outperforms WJ in this case.

\begin{figure*}[!ht]
        \centering
        \vspace*{-1ex}
        \begin{subfigure}[b]{0.245\textwidth}
            \includegraphics[width=\textwidth,clip=true]{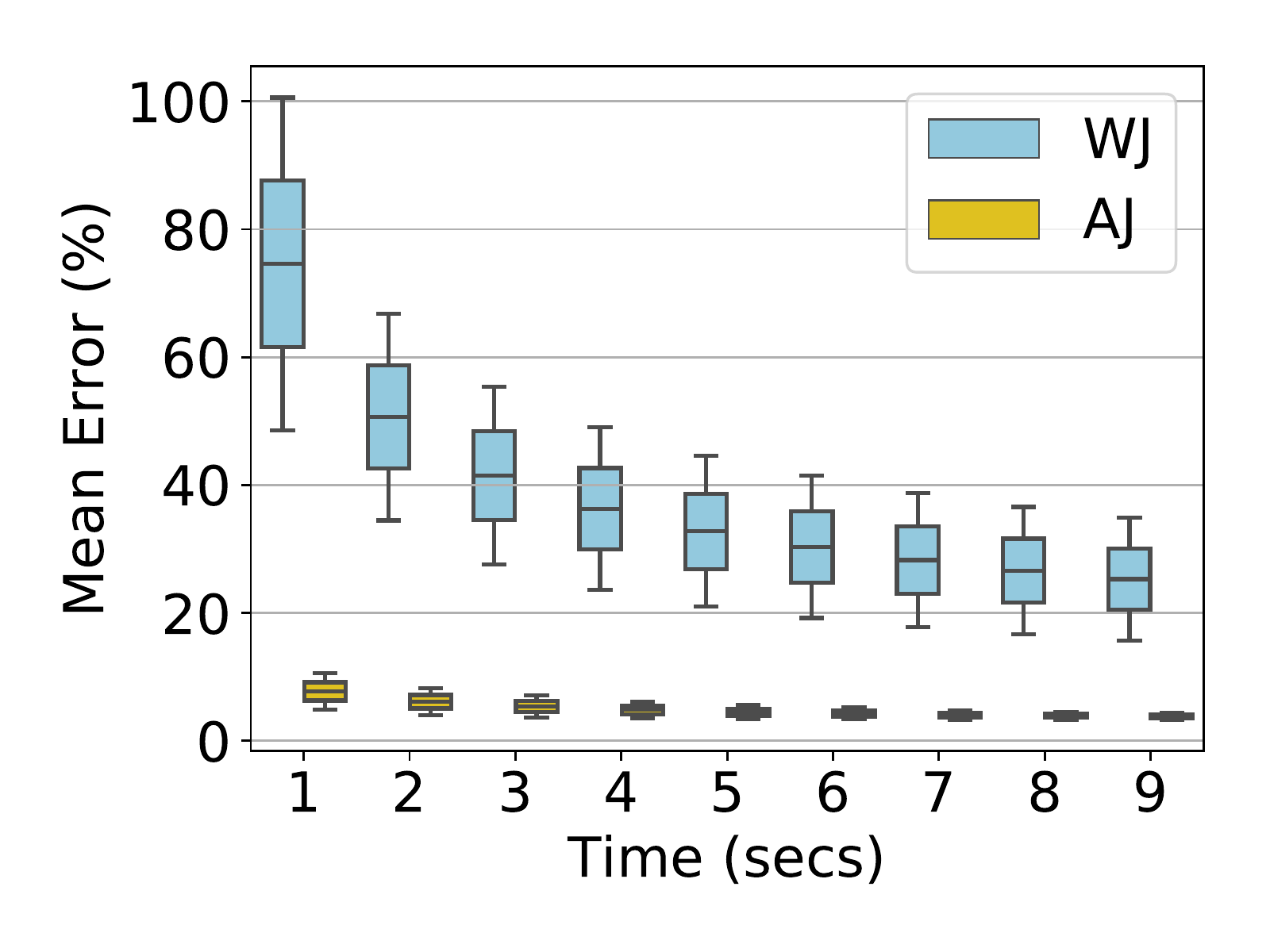}\vspace{-2ex}
            \caption{Step one queries \label{fig:whiskerBag1}}
        \end{subfigure}
        \begin{subfigure}[b]{0.245\textwidth}
            \includegraphics[width=\textwidth,clip=true]{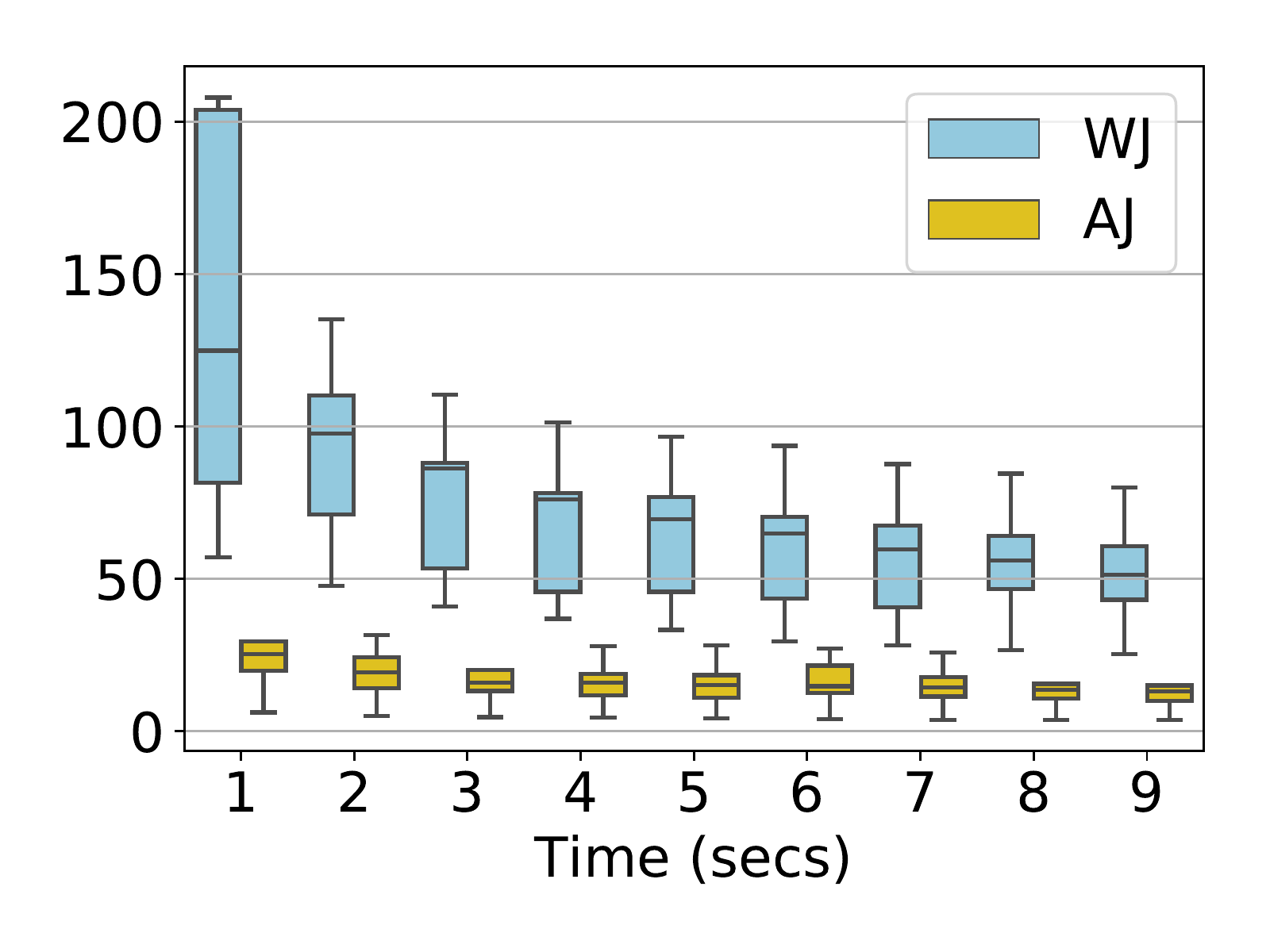}\vspace{-2ex}
            \caption{Step two queries \label{fig:whiskerBag2}}
        \end{subfigure}
        \begin{subfigure}[b]{0.245\textwidth}
            \includegraphics[width=\textwidth,clip=true]{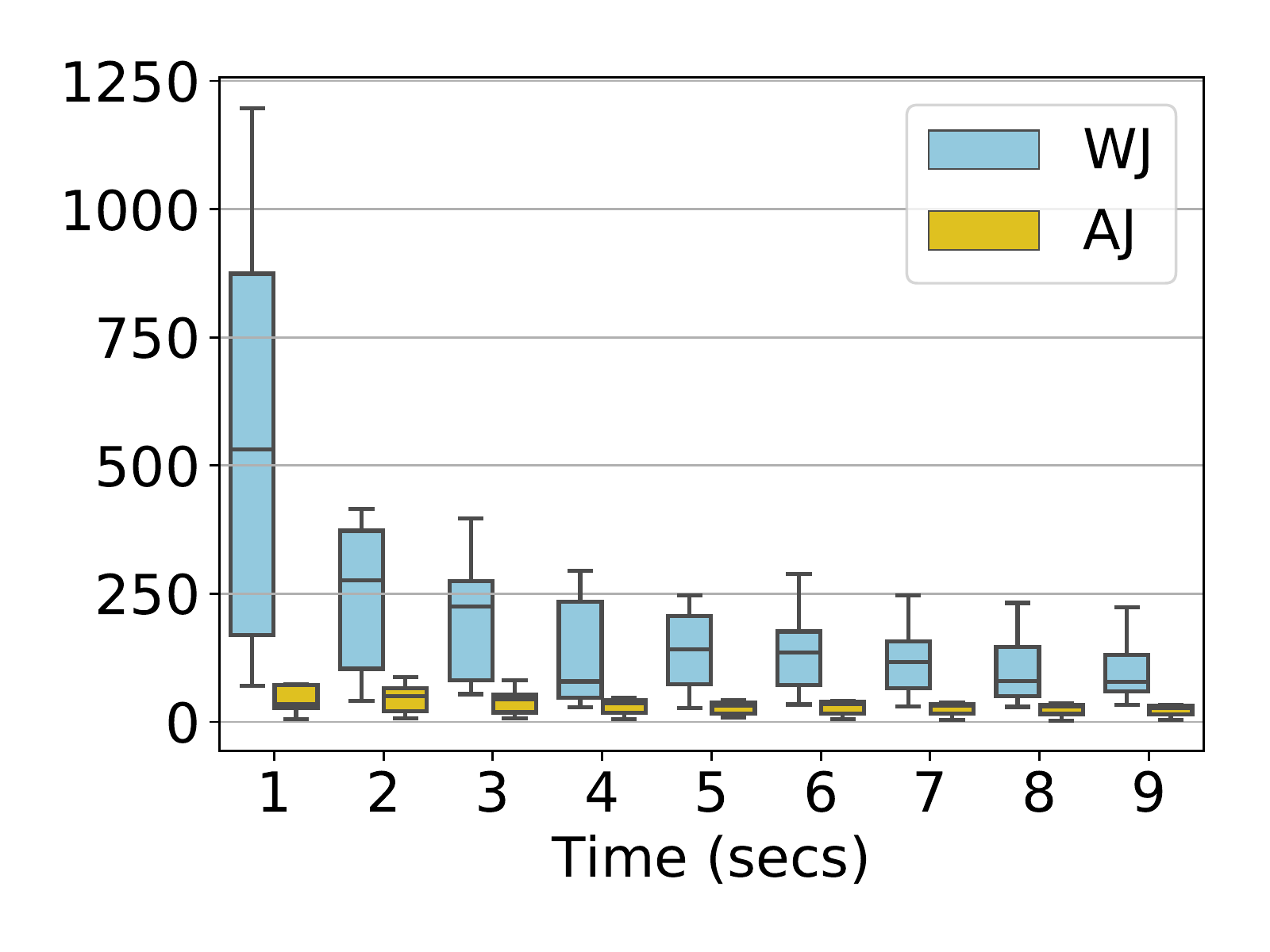}\vspace{-2ex}
            \caption{Step three queries \label{fig:whiskerBag3}}
        \end{subfigure}
        \begin{subfigure}[b]{0.245\textwidth}
            \includegraphics[width=\textwidth,clip=true]{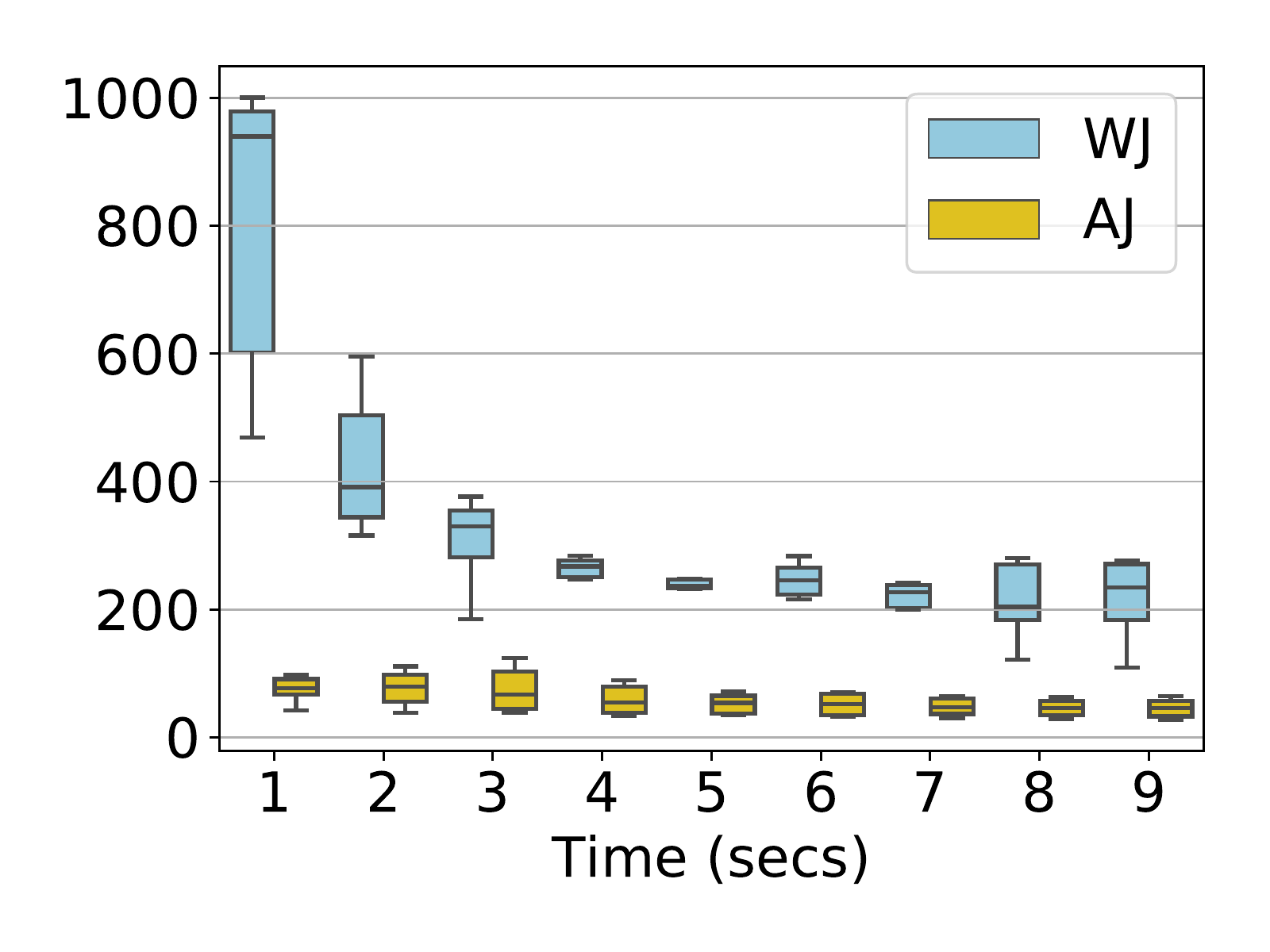}\vspace{-2ex}
            \caption{Step four queries \label{fig:whiskerBag4}}
        \end{subfigure}
        \caption{Error over time in seconds for queries without the distinct operator  \label{fig:whiskerBag}}
\end{figure*}
%% Whisker plots %%

% graph explanation
\paragraph{All queries with distinct:}
For each exploration depth and knowledge graph, Figure~\ref{fig:whisker} presents the range of mean errors for all estimations under the distinct operator as box plots (specifically, Tukey plots), displaying the interquartile range of error (the box), the median error (the line in each box), and the most extreme data point within 1.5$\times$ the interquartile range (the whiskers). These plots show the variation of mean error across all queries with respect to time.

% discuss WJ performance compared to AJ
These results show that AJ is consistently better than WJ over all randomly generated queries: the median of mean errors for WJ in some cases reaches over 1,000\% after one second and 300\% after nine seconds (see Figure~\ref{fig:whisker4d}); on the other hand, the same median errors for AJ are at worst 104\% after one second, and 50\% after nine seconds. 

% For example, Figure~\ref{fig:whisker1d} presents the generated queries of step one in the random exploration paths of DBpedia. In the figure, the variation in the mean error of the estimation is shown as the size of the boxes and whiskers. Following these ranges, it is clear that AJ consistently outperforms WJ. Similar results can be found in Figure~\ref{fig:whisker1l} for the step one queries on LinkedGeoData. 

% Comparison of depth
An interesting result can be found when comparing the graphs for a varying numbers of exploration steps. Taking LinkedGeoData, for example, the accuracy of WJ drops when comparing the first step (Figure~\ref{fig:whisker1l}) with 
subsequent steps (Figures~\ref{fig:whisker2l}--\ref{fig:whisker4l} resp.). A similar trend can be seen for DBpedia (though less clear moving from step one to two). We attribute this trend to \e{(1)} an increasing number of rejections: later steps tend become more specific, adding more selective joins; \e{and (2)} increasing duplicates: as the query becomes larger, more variables are projected away. Both issues are specifically addressed by the AJ algorithm.

\paragraph{All queries without distinct:}
Finally, though our exploration system requires the distinct operator for counts, we also perform experiments in the non-distinct case to understand the relative impact of the unbiased distinct estimator and the partial exact computations on reducing error in AJ. Figure~\ref{fig:whiskerBag} presents Tukey plots of mean error over time for all queries without distinct, and both datasets, separated by the number of exploration steps. Overall we conclude that: \e{(1)} the error observed for WJ drops from the distinct case to the non-distinct case; \e{(2)} the errors generally increase in AJ versus the distinct case; \e{and (3)} AJ continues to significantly outperform WJ, though by a lesser margin when considering the non-distinct case. We surmise from these results that though AJ no longer benefits from the unbiased distinct estimator, it continues to clearly outperform WJ in the non-distinct case due to the partial exact computations: the benefits of AJ are not only due to the unbiased distinct estimator.

\paragraph{Summary:} Our first results show that computing exact counts for a selection of exploration queries may take in the order of hours for Virtuoso, and in the order of tens of seconds for CTJ, with both approaches slowing down for the larger LinkedGeoData graph. Though CTJ offers a major performance benefit versus Virtuoso, we conclude that computing exact results (with either approach) is not compatible with our goal of interactive runtimes. Focusing thereafter on online aggregation, we find that AJ significantly reduces error (by orders of magnitude) versus WJ over the same time period, and in most cases can quickly converge to estimates with less than $1\%$ mean error. We also found that the performance of the online aggregation approaches is less affected by the larger size of LinkedGeoData than in the exact settings, and that the relative benefit of AJ and WJ improves as more exploration steps are added. Finally, through experiments on non-distinct queries, we find that AJ continues to significantly reduce error versus WJ due to its inclusion of CTJ for partial exact computations.

% Figure~\ref{fig:whiskerBag1} shows all the queries of step one. The results show that the accuracy of WJ improved compared to estimator with the DISTINCT. For example, The median mean error of WJ after one second is 75\%, which is lower than the median of WJ with the DISTINCT which was 100\%. While the accuracy of \audit drops, it is still almost an order of magnitude better than WJ, with a median of 8\%. Also, Comparing the variance in mean error between the two solutions, shows that \audit offers consistently better results without the DISTINCT estimator.

% Figure~\ref{fig:whiskerBag2} and Figure~\ref{fig:whiskerBag3} shows the results of bag semantic queries for steps two and three respectively. When comparing to step one, we can see similar results of increasing mean error for both algorithms when advancing in steps. For example, the median mean error of WJ after one second in step two is 125\% and more than 288\% in step three. Similar results can be found for \audit, with median mean error of 25\% and 41\% in steps two and three, respectively. This behavior is similar to the results with the DISTINCT estimator and caused by the exact partial aggregation decision.

	%\vspace{-1ex}
\section{Conclusions}
We presented a formal framework for the visual exploration of knowledge
graphs, where an exploration step consists of the transformation of
the bar of one bar chart into the next bar chart. We also investigated
the implementation of this framework using various query engines for
rendering bar charts, including a SPARQL engine, a recent in-memory
join algorithm, and online aggregation. Finally, we devised and
analyzed Audit Join: a specialized online-aggregation
algorithm that combines the random walks of Wander Join with exact
computation, and extends its estimator to accommodate the
count aggregations under the distinct operator. Our experiments show that the runtimes of both methods for performing exact counts are too slow for supporting interactive exploration over large-scale knowledge graphs. Focusing thereafter on online aggregation, we find that when compared with Wander Join, Audit Join significantly reduces error with respect to time in all experiments, often by orders of magnitude, including both distinct and non-distinct cases.

In terms of future directions, for one, we wish to arrive at a more general understanding of how partial exact computation complements online aggregation in order to devise a more principled way to combine both; our results, along with those of Zhao et al.~\cite{Zhao:2018:RSO:3183713.3183739}, show this to be a very promising approach in general, and more work can be done to refine this idea. We will also explore the application of our approach to general join queries, beyond the path queries produced by our system. We also wish to explore directions for improving the usability of our exploration system, which include: allowing to explore and contrast multiple knowledge graphs simultaneously, adding support for incremental indexing on updates, extending filtering capabilities, and adding support for further semantics beyond subclass closure. Finally, we plan to conduct user studies to increase the potential of our system for exploring large-scale knowledge graphs.

	\balance
	%\section{Acknowledgment}

%\newpage 

	\bibliographystyle{abbrv}
	\bibliography{ref} 
	\appendix
\def\query#1{#1}
\def\expname#1{\textit{#1}}
\def\query#1{\small #1}
\def\subclss#1{\textsf{\expname{SubC}[#1]}}
\def\inProp#1{\textsf{\expname{InProp}[#1]}}
\def\outProp#1{\textsf{\expname{OutProp}[#1]}}
\def\objects#1{\textsf{\expname{Obj}[#1]}}
\def\subjects#1{\textsf{\expname{Sbj}[#1]}}
\def\path#1#2#3#4#5{#2 -> #3 -> #4 -> #5 \\}

\def\ra{\rightarrow}
\def\path#1#2#3#4#5{#2 $\ra$ #3 $\ra$ #4 $\ra$ #5}
\def\paththree#1#2#3#4{#2 $\ra$ #3 $\ra$ #4}
\def\pathtwo#1#2#3{#2 $\ra$ #3}
\def\pathone#1#2{#2}

\section{Random Queries}
In this section, we list the queries produced by our random generator
for the experiments.

\subsection{DBPedia}
 We use the following namespace acronyms:

\begin{itemize}\itemsep0.5em 
\item \textsf{rdf}: \textsf{http://www.w3.org/1999/02/22-rdf-syntax-ns}
\item \textsf{owl}: \textsf{http://www.w3.org/2002/07/owl}
\item \textsf{dbo}: \textsf{http://dbpedia.org/ontology}
\item \textsf{dbp}: \textsf{http://dbpedia.org/property}
\item \textsf{foaf}: \textsf{http://xmlns.com/foaf/0.1}
\item \textsf{yago}: \textsf{http://dbpedia.org/class/yago}
\end{itemize}
The exploration paths are the following. Each step is a query, and we disregard duplicates of queries in our experiments.
\begin{enumerate} \itemsep0.5em 
\item \pathtwo{dbpedia}{\query{
\inProp{owl:Thing}}}{\query{\subjects{dbo:wikiPageRedirects}}} 
\item \pathtwo{dbpedia}{\query{ \inProp{owl:Thing}}}{\query{\subjects{foaf:primaryTopic}}}
\item \path{dbpedia}{\query{ \outProp{owl:Thing} }}{\query{ \objects{dbo:musicalArtist} }\\}{\query{ \subclss{dbo:Person} }}{\query{ \subclss{dbo:Artist} }}
\item \path{dbpedia}{\query{ \outProp{owl:Thing} }}{\query{ \objects{dbo:musicalArtist} }\\}{\query{ \subclss{dbo:Person} }}{\query{ \subclss{dbo:Artist} }}
\item \pathtwo{dbpedia}{\query{ \outProp{owl:Thing} }}{\query{ \objects{rdf:type} }}
\item \pathtwo{dbpedia}{\query{ \subclss{owl:Thing} }}{\query{ \inProp{dbo:Agent} }}
\item \pathtwo{dbpedia}{\query{ \subclss{owl:Thing} }}{\query{ \inProp{dbo:TimePeriod} }}
\item \pathtwo{dbpedia}{\query{ \subclss{owl:Thing} }}{\query{ \outProp{dbo:Agent} }}

\item \paththree{dbpedia}{\query{ \subclss{owl:Thing} }}{\query{ \outProp{dbo:MeanOfTransportation} }}{\query{ \objects{dbo:length} }}
\item \pathtwo{dbpedia}{\query{ \subclss{owl:Thing} }}{\query{ \outProp{dbo:Species} }}
\item \paththree{dbpedia}{\query{ \subclss{owl:Thing} }}{\query{ \outProp{dbo:TimePeriod}\\ }}{\query{ \objects{rdf\#type} }}

\item \paththree{dbpedia}{\query{ \subclss{owl:Thing} }}{\query{ \outProp{dbo:Work} \\}}{\query{ \objects{dbp:background} }}
\item \paththree{dbpedia}{\query{ \subclss{owl:Thing} }}{\query{ \outProp{dbo:Work} \\}}{\query{ \objects{dbp:distributor} }}

\item \path{dbpedia}{\query{ \subclss{owl:Thing} }}{\query{ \outProp{dbo:Work}\\ }}{\query{ \objects{dbp:guests} }}{\query{ \subclss{yago:Person100007846} }}
\item \paththree{dbpedia}{\query{ \subclss{owl:Thing} }}{\query{ \outProp{dbo:Work} }}{\query{ \objects{rdf\#type} }}

\item \path{dbpedia}{\query{ \subclss{owl:Thing} }}{\query{ \subclss{dbo:Place}\\ }}{\query{ \subclss{dbo:ArchitecturalStructure}\\ }}{\query{ \outProp{dbo:Infrastructure} }}
\item \path{dbpedia}{\query{ \subclss{owl:Thing} }}{\query{ \subclss{dbo:Place} \\}}{\query{ \subclss{dbo:NaturalPlace} }}{\query{ \outProp{dbo:Volcano} }}

\item  \pathtwo{dbpedia}{\query{ \subclss{owl:Thing} }}{\query{ \subclss{dbo:TimePeriod}\\ }}
\item  \path{dbpedia}{\query{ \subclss{owl:Thing} }}{\query{ \subclss{dbo:Work}\\ }}{\query{ \outProp{dbo:Document} }}{\query{ \objects{rdf:type} }}
\end{enumerate}

\subsection{LinkedGeoData}
 We use the following namespace acronym:
 \begin{itemize}\itemsep0.5em 
\item \textsf{lgdo}: \textsf{http://linkedgeodata.org/ontology/}
\end{itemize}

The exploration paths are the following. Each step is a query, and we disregard duplicates of queries in our experiments. \inProp{$\cdot$} expansions of LinkedGeoData always returned an empty result.
\begin{enumerate} \itemsep0.5em
\item \pathone{ linkedgeodata }{\query{ \outProp{ConnectTrees} }}
\item \pathtwo{ linkedgeodata }{\query{ \subclss{ConnectTrees} }}{\query{ \outProp{lgdo:BarrierThing} }}
\item \pathtwo{ linkedgeodata }{\query{ \subclss{ConnectTrees} }}{\query{ \outProp{lgdo:ManMadeThing} }}
\item \pathtwo{ linkedgeodata }{\query{ \subclss{ConnectTrees} }}{\query{ \outProp{lgdo:Place} }}
\item \path{ linkedgeodata }{\query{ \subclss{ConnectTrees} }}{\query{ \subclss{lgdo:Amenity}\\ }}{\query{ \subclss{lgdo:Shop} }}{\query{ \outProp{lgdo:Convenience} }}
\item \path{ linkedgeodata }{\query{ \subclss{ConnectTrees} }}{\query{ \subclss{lgdo:Amenity}\\ }}{\query{ \subclss{lgdo:Shop} }}{\query{ \outProp{lgdo:Craft} }}
\item \path{ linkedgeodata }{\query{ \subclss{ConnectTrees} }}{\query{ \subclss{lgdo:Amenity}\\ }}{\query{ \subclss{lgdo:Shop} }}{\query{ \outProp{lgdo:Kiosk} }}
\item \path{ linkedgeodata }{\query{ \subclss{ConnectTrees} }}{\query{ \subclss{lgdo:Amenity}\\ }}{\query{ \subclss{lgdo:Shop} }}{\query{ \outProp{lgdo:Pharmacy} }}
\item \path{ linkedgeodata }{\query{ \subclss{ConnectTrees} }}{\query{ \subclss{lgdo:Amenity}\\ }}{\query{ \subclss{lgdo:Shop} }}{\query{ \outProp{lgdo:Supermarket} }}
\item \paththree{ linkedgeodata }{\query{ \subclss{ConnectTrees} }}{\query{ \subclss{lgdo:BarrierThing} \\}}{\query{ \outProp{lgdo:Fence} }}
\item \pathtwo{ linkedgeodata }{\query{ \subclss{ConnectTrees} }}{\query{ \subclss{lgdo:BarrierThing}\\ }}
\item \paththree{ linkedgeodata }{\query{ \subclss{ConnectTrees} }}{\query{ \subclss{lgdo:Place} \\}}{\query{ \outProp{lgdo:Village} }}
\item \pathtwo{ linkedgeodata }{\query{ \subclss{ConnectTrees} }}{\query{ \subclss{lgdo:PowerThing}\\ }}
\item \paththree{ linkedgeodata }{\query{ \subclss{ConnectTrees} \\}}{\query{ \subclss{lgdo:RailwayThing} }}{\query{ \outProp{lgdo:BufferStop} }}
\end{enumerate}

\end{document}